\documentclass{article}

 \usepackage{PRIMEarxiv}

\usepackage[hidelinks]{hyperref}
\usepackage{booktabs}
\usepackage{amsfonts}
\usepackage{float}
\usepackage{dsfont}
\usepackage{algorithm}
\usepackage{algpseudocode}
\usepackage[citestyle=authoryear,natbib=true,backend=bibtex,doi=false,isbn=false,url=false,eprint=false]{biblatex}
\usepackage{cancel}

\usepackage{wrapfig}
\usepackage{float}
\usepackage{bbm}
\usepackage{tabularx}
\usepackage{lineno}
\usepackage{changepage}
\usepackage{adjustbox}
\usepackage{cancel}
\usepackage{amsmath,amsthm,amssymb,amsfonts, fancyhdr, color, comment, graphicx, environ}
\usepackage{mathtools}
\usepackage{amsmath}

\newcommand{\indep}{\perp \!\!\! \perp}
\DeclareMathOperator*{\argmax}{arg\,max}
\DeclareMathOperator*{\argmin}{arg\,min}
\DeclareMathOperator*{\Argmax}{Arg\,max}

\DeclarePairedDelimiter\ceil{\lceil}{\rceil}

\usepackage[capitalize,noabbrev]{cleveref}

\theoremstyle{plain}
\newtheorem{theorem}{Theorem}[section]

\newtheorem{lemma}[theorem]{Lemma}
\newtheorem{corollary}[theorem]{Corollary}
\theoremstyle{definition}
\newtheorem{definition}[theorem]{Definition}
\newtheorem{assumption}[theorem]{Assumption}
\theoremstyle{remark}

\algtext*{EndFor}
\algrenewcommand\algorithmicindent{1.0em}

\title{
Conformal Robust Control of Linear Systems
}

\author{
  Yash Patel \\
  Department of Statistics \\
  University of Michigan \\
  Ann Arbor, MI 48104 \\
  \texttt{yppatel@umich.edu}
  \And
  Sahana Rayan \\
  Department of Statistics \\
  University of Michigan \\
  Ann Arbor, MI 48104 \\
  \texttt{srayan@umich.edu}
  \And
  Ambuj Tewari \\
  Department of Statistics \\
  University of Michigan \\
  Ann Arbor, MI 48104 \\
  \texttt{tewaria@umich.edu}
}

\begin{document}

\maketitle

%

%





\begin{abstract}
 End-to-end engineering design pipelines, in which designs are evaluated using concurrently defined optimal controllers, are becoming increasingly common in practice. To discover designs that perform well even under the misspecification of system dynamics, such end-to-end pipelines have now begun evaluating designs with a robust control objective. Current approaches of specifying such robust control subproblems, however, rely on hand specification of perturbations anticipated to be present upon deployment or margin methods that ignore problem structure, resulting in a lack of theoretical guarantees and overly conservative empirical performance. We, instead, propose a novel methodology for LQR systems that leverages conformal prediction to specify such uncertainty regions in a data-driven fashion. Such regions have distribution-free coverage guarantees on the true system dynamics, in turn allowing for a probabilistic characterization of the regret of the resulting robust controller. We then demonstrate that such a controller can be efficiently produced via a novel policy gradient method that has convergence guarantees. We then demonstrate the superior empirical performance of our method over alternate robust control specifications, such as $\mathcal{H}_{\infty}$ and LQR with multiplicative noise, across a collection of engineering tasks.
\end{abstract}

\section{Introduction}\label{section:intro}

Seeking control over a family of dynamical systems is a problem often encountered in engineering \citep{killian2016transfer,wu2018learning,aksland2023approach}. One prevalent application of this is in cases where engineering designs and their respective controllers are being concurrently developed, known as control co-design (CCD) \citep{garcia2019control}. Traditional engineering design loops operated sequentially, first proposing a design and then developing a controller \citep{reyer2001comparison,friedland1995advanced}. Such workflows, however, sacrificed the improved optimality possible in their coupling, hence the increasing interest in leveraging end-to-end co-control design pipelines \citep{fathy2001coupling,falck2021dymos}. 

Initial works in CCD studied optimal design assuming perfectly specified, deterministic system dynamics \citep{allison2014co,azad2018comprehensive,azad2019phev,behtash2018decomposition}. Such assumptions have, however, become overly restrictive, resulting in interest in robust extensions of the CCD formulation, referred to as uncertain CCD (UCCD) \citep{azad2022control,azad2023overview,bird2023set}. Such uncertainty can arise from many sources in the design process, such as noise in the controllers, uncertainties in the design parameters, or unmodeled dynamics. The UCCD specification also differs depending on the risk tolerance in the downstream application. For instance, in risk-neutral settings, stochastic specifications are appropriate \citep{azad2020single,cui2022control,behtash2024comparative}, whereas in risk-averse settings, probabilistic \citep{cui2020comparative,cui2020reliability,nguyen2022reliability} or worst-case forms \citep{azad2021robust,nash2021robust,azad2020robust} are used. 

We focus on the worst-case robust UCCD formulation (WCR-UCCD), specifically on dynamics misspecification. WCR-UCCD requires specifying a dynamics uncertainty region. Existing methods of specification, however, tend to be ad-hoc and, thus, fail to provide any guarantees of the robust solution as it relates to the selection of this uncertainty set, rendering its choice often difficult and resulting in suboptimal controller synthesis \citep{azad2020single}.

We, thus, focus herein on providing a principled distribution-free specification of the \textit{robust control subproblem} in WCR-UCCD and an associated solution method with convergence guarantees. One special case of interest in UCCD is in the setting of linear quadratic regulators (LQRs), where the underlying system dynamics take on a linear structure \citep{ahmadi2023lqr,fathy2003nested,jiang2016iterative}. LQR systems are of broad interest both due to their analytic tractability and widespread applicability to practical engineering systems \citep{zhao2024kalman,mamakoukas2019local,bevanda2022towards}. We, therefore, propose a method for specifying the LQR WCR-UCCD control subproblem that lends itself to efficient solution by leveraging conformal prediction on observed design information. 
A related use of conformal prediction for predict-then-optimize problems was recently studied in \citep{patel2024conformal}. Unlike their setting, however, the application of conformal prediction to control has complications related to the stability of controllers under model uncertainty. Our contributions are as follows:

\begin{itemize}
    \item Providing a framework to define robust LQR control problems with distribution-free probabilistic regret guarantees, across deterministic or stochastic and discrete- or continuous-time dynamics, and demonstrating empirical improvements over alternative robust control schemes.
    \item Extending conformalized predict-then-optimize to cases where calibration data is observed with noise and where the domains of both the maximization \textit{and} minimization components of the robust formulation depend on the conformalized predictor.
    \item Providing a novel policy subgradient method for robust controller synthesis with convergence guarantees proven via subgradient dominance.
\end{itemize}

\section{Background}\label{section:background}

\subsection{Conformal Prediction}
Conformal prediction is a principled, distribution-free approach of uncertainty quantification \citep{angelopoulos2021gentle, shafer2008tutorial}. ``Split conformal,'' the most common variant of conformal prediction, is used as a wrapper around predictors $\widehat{f} : \mathcal{X}\rightarrow\mathcal{Y}$ such that prediction \textit{regions} $\mathcal{C}(x)$ are returned in place of the point predictions $\widehat{f}(x)$. Prediction regions $\mathcal{C}(x)$ are sought to have coverage guarantees on the true $y := f(x)$. That is, for some prespecified $\alpha$, we wish to have $\mathcal{P}_{X,Y}(Y\in \mathcal{C}(X))\ge1-\alpha$.

To achieve this, split conformal partitions the overall dataset $\mathcal{D}$ into two subsets, $\mathcal{D}_{T}\cup\mathcal{D}_{C}$, respectively the training and calibration datasets. After fitting $\widehat{f}$ on the training subset, the calibration set is used to measure the anticipated ``prediction error'' for future test points. Formally, this error is quantified via a score function $s(x,y)$, which generalizes the classical notion of a residual. In particular, scores are evaluated on the calibration dataset to define $\mathcal{S} := \{s(x,y)\mid (x,y)\in\mathcal{D}_{C}\}$. Denoting the $\ceil{(|\mathcal{D}_{C}|+1)(1-\alpha)}/|\mathcal{D}_{C}|$ empirical quantile of $\mathcal{S}$ as $\widehat{q}$, conformal prediction defines $\mathcal{C}(x)$ to be $\{y \mid s(x, y) \le \widehat{q}\}$. Such $\mathcal{C}(x)$ satisfies the aforementioned coverage guarantees under the exchangeability of future test points $(x',y')$ with points from $\mathcal{D}_{C}$. While the coverage guarantee holds for any $s(x,y)$, the sizes of the resulting prediction regions, known as the procedure's ``predictive efficiency,'' is dependent on its choice \citep{shafer2008tutorial}. 

\subsection{Robust Predict-Then-Optimize}\label{section:bg_pred_opt}
Predict-then-optimize problems are nominally formulated as $w^{*}(x) := \min_{w\in\mathcal{W}} \mathbb{E}[f(w, C)\mid x],$ where $w$ are decision variables, $C$ an \textit{unknown} cost parameter, $x$ observed contextual variables, $\mathcal{W}$ a compact feasible region, and $f(w, c)$ an objective function that is convex-concave and $L$-Lipschitz in $c$ for any fixed $w$. The nominal approach defines a predictor $\widehat{g} : \mathcal{X}\rightarrow\mathcal{C}$, where the prediction $\widehat{c} := \widehat{g}(x)$ is leveraged, i.e. taking $w^{*}(x) := \min_{w} f(w, \widehat{c})$. 
Such an approach, however, is inappropriate in safety-critical settings, given that $\widehat{g}$ will likely be misspecified and, thus, result in decisions that are suboptimal under the true cost parameter, $c$. For this reason, robust alternatives to the nominal formulation have become of interest \citep{chenreddy2022data,sadana2024survey,chenreddy2024end}. We focus on the following formulation from \citep{patel2024conformal}:
\begin{equation}\label{eqn:reform_obj}
w^{*}(x) := \min_{w} \max_{\widehat{c}\in\mathcal{U}(x)} \quad f(w, \widehat{c})
\end{equation}
where $\mathcal{U} : \mathcal{X}\rightarrow\mathcal{F}$ is a uncertainty region predictor, with $\mathcal{F}$ being the $\sigma$-field of $\mathcal{C}$, such that $\mathcal{P}_{X,C}(C\in\mathcal{U}(X)) \ge 1-\alpha$. Works in this field typically study the suboptimality gap, defined as
$\Delta(x, c) := \min_{w} \max_{\widehat{c}\in\mathcal{U}(x)} f(w, \widehat{c}) - \min_{w} f(w, c)$. 
For instance,
in \citep{patel2024conformal}, $\mathcal{U}(x)$ was specifically constructed via conformal prediction to provide probabilistic guarantees; that is, by taking $\mathcal{U}(x) := \mathcal{C}(x)$ to be the prediction region produced by conformalizing the predictor $\widehat{g}$, they demonstrated $\mathcal{P}_{X,C}\left(0\le \Delta(X, C) \le L \mathrm{\ diam}(\mathcal{U}(X))\right) \ge 1 - \alpha$. 

\subsection{LQR \& Control Co-Design}\label{section:lqr_ccd_bg}   
The field of control has a long history in engineering physics and robotics \citep{zabczyk2020mathematical}. In the linear quadratic regulator (LQR) setup, the state dynamics have a linear form $\dot{x} = Ax + Bu + w$, where $x$ is the state, $u$ the control inputs, and $w\sim\mathcal{D}_{w}$ the noise. Optimal control is then posed as an optimization problem, with the objective $J(u)$ weighing both the deviation from a target state and the necessary control input. LQR optimal controllers take a linear feedback form, namely $u^{*}(x) = -K^{*}x$ where $K^{*}$ is known as the ``optimal gain matrix'' and solves
\begin{align}\label{eqn:lqr}
K^{*}(A, B) &:= \argmin_{K\in\mathcal{K}(A, B)} \mathbb{E}[J(K, A, B)] \\
\quad \mathrm{where\ } J(K, A, B) &:= \int_{0}^{\infty} (x^{\top} Q x + (Kx)^{\top} R (Kx)) dt \nonumber
\end{align}
where $\mathcal{K}(A,B) := \{K : \mathrm{Re}(\lambda_i(A - BK)) < 0\ \forall i\}$ is known as the set of ``stabilizing controllers'' for the $A,B$ system dynamics and $\dot{x} = (A - BK)x + w$.
Variations, where the integral is replaced by a discretized sum or considered to some finite $T$, are also of interest. Solving this problem is often done either by solving the algebraic Ricatti equation (ARE) \citep{willems1971least} or via policy gradient \citep{sun2021learning}. 

We now briefly summarize the relevant pieces of uncertain co-control design to highlight the robust control problem subproblem therein; for a full survey, refer to \citep{azad2022control}. Engineering designs can often be specified by parameters $\theta$, which could capture, for instance, the dimensions of an airfoil or material properties of a DC battery grid. The dynamics are highly dependent on the design; for example, an airfoil with a shape $\theta_{1}$ will fly differently from one given by $\theta_{2}$. Worst-case robust UCCD with dynamics misspecification, thus, solves
$\min_{K,\theta} \max_{\widehat{A}\in\mathcal{A}(\theta),\widehat{B}\in\mathcal{B}(\theta)} \mathbb{E}[o(K, \widehat{A}, \widehat{B}, \theta)]$, where $\dot{x} = \widehat{A}x + \widehat{B}u + w$ and
$(\mathcal{A}(\theta),\mathcal{B}(\theta))$ are uncertainty sets of the dynamics for such a design and $o$ is the objective.  

Often, the objective takes a decomposable form, namely with one term relating to system control and the other depending on the design parameter, i.e. $o(K, A, B, \theta) := \ell(\theta) + J(K, A(\theta), B(\theta))$ \citep{chanekar2018co,ahmadi2023lqr}. 
One commonly applied solution technique in this setting is via bilevel optimization, in which an
outer optimization loop is performed over design parameters and an inner one over controllers for the current design iterate \citep{herber2019nested,kamadan2017co}. For this reason, the specification of the robust control subproblem can be studied independently of the outer design optimization loop, as done herein. 



\section{Methodology}\label{section:method}

We now discuss conformally robust LQR, providing the formulation in \Cref{section:problem_formulation}, regret guarantees in \Cref{section:cov_guar} and \Cref{section:ambiguous_ground_truth}, and a controller synthesis algorithm with convergence guarantees in \Cref{section:optimization}. 

\subsection{Problem Formulation}\label{section:problem_formulation}
For the presentation below, let $x_{t}\in\mathbb{R}^{n}$, $u_{t}\in\mathbb{R}^{m}$, $A\in\mathbb{R}^{n\times n}$, and $B\in\mathbb{R}^{n\times m}$. Let $C$ denote the full dynamics matrix $C := [A, B]\in\mathbb{R}^{n\times(n+m)}$. We additionally assume a linear control scheme, namely $u_{t} = -K x_{t}$ for some gain matrix $K$. Additionally, denote $W := [I_{n\times n} \mathrm{\ -}K^{\top}]^{\top}\in\mathbb{R}^{(n+m)\times n}$, such that the closed-loop dynamics are given by $CW = A - BK$. 
As discussed in \Cref{section:lqr_ccd_bg}, we assume a dataset of designs and associated trajectories is observed. We assume such a dataset $\mathcal{D}$ consists of $N$ samples $(\theta^{(i)}, C^{(i)})\sim\mathcal{P}(\Theta,C)$ and $K^{(i)}\sim\mathcal{P}(K)$, where $\mathcal{P}(\Theta,C)$ is an unknown joint distribution over designs and dynamics and $\mathcal{P}(K)$ an unknown distribution on gain matrices. We make no assumptions on such distributions other than that each gain matrix $K^{(i)}$ is a stabilizing controller for the respective $C^{(i)}$ dynamics. Note that these underlying true dynamics $C^{(i)}$ are never observed directly by the learning algorithm; only the resulting trajectories are observed. Such trajectories are generated by evolving the state via $x_{t+1}^{(i)} = (C^{(i)} W^{(i)}) x_{t}^{(i)}$ over a time horizon $T$. 
The final dataset, therefore, takes the form $\mathcal{D} = \{\theta^{(i)}, \{(x_{t}^{(i)}, u_{t}^{(i)})\}_{t=1}^{T}\}_{i=1}^{N}$.
We are interested in studying a risk-sensitive formulation of LQR:
\begin{equation}\label{eqn:robust_objective}
\begin{gathered}
    K^{*}_{\mathrm{rob}}(\mathcal{U}(\theta)) := \argmin_{K\in\mathcal{K}(\mathcal{U}(\theta))} \max_{[\widehat{A},\widehat{B}]:=\widehat{C}\in\mathcal{U}(\theta)} \mathbb{E}[J(K, \widehat{A}, \widehat{B})] \\
    \textrm{s.t.} \quad \dot{x} = \widehat{A}x + \widehat{B}u + w
    \qquad \mathcal{P}_{\Theta,C}(C \in \mathcal{U}(\Theta)) \ge 1 - \alpha \nonumber,    
\end{gathered}
\end{equation}
where $J$ is the objective function particular to the setting of interest, differing between infinite and finite time horizons and continuous and discrete time dynamics, and $\mathcal{U}(\theta)$ is an uncertainty set over dynamics. Notably, the notion of stabilizing controllers must be generalized in this robust formulation, since the nominal formulation is for a specific $C$. We, thus, consider those controllers that stabilize the entire uncertainty set, which we refer to as the ``universal stabilizing set,'' formally $\mathcal{K}(\mathcal{U}(\theta)) := \bigcap_{\widehat{C}\in\mathcal{U}(\theta)} \mathcal{K}(\widehat{C})$, where $\mathcal{K}(\widehat{C})$ is \Cref{eqn:lqr} evaluated for a particular $\widehat{A},\widehat{B}$. 

\subsection{Score Function}\label{section:score}

From the trajectories in $\mathcal{D}$, we can perform system identification using least squares estimation to recover estimates of the system dynamics, $(\widetilde{A}^{(i)}, \widetilde{B}^{(i)})$ \citep{ljung1987theory}. With this, we obtain a final dynamics dataset $\widetilde{\mathcal{D}} = \{\theta^{(i)}, \widetilde{C}^{(i)}\}_{i=1}^{N}$, which we then leverage in the standard manner of split conformal prediction. That is, we split $\widetilde{\mathcal{D}} = \widetilde{\mathcal{D}}_{\mathcal{T}} \cup \widetilde{\mathcal{D}}_{\mathcal{C}}$, the former of which we use to train a system parameters predictor $\widehat{C} := f(\theta)$. Notably, leveraging split conformal in this setting has the complication that the ground truth used, namely in $\widetilde{\mathcal{D}}_{\mathcal{C}}$, is itself an estimate $\widetilde{C}$ even though coverage is sought on $C$. We assume for this initial discussion that for a fixed coverage level $\alpha$, we can obtain prediction regions with the desired coverage, satisfying $\mathcal{P}_{\Theta,C}(C\in\mathcal{U}(\Theta))\ge 1-\alpha$, using $\widetilde{\mathcal{D}}_{\mathcal{C}}$. The treatment of this gap between $\widetilde{C}$ and $C$ is discussed in \Cref{section:ambiguous_ground_truth}.

We take the score to be 
$s(\theta, C) = ||f(\theta) - C ||_{\mathrm{op}}$,
where $||\cdot||_{\mathrm{op}}$ is the matrix \textit{operator} norm, i.e. $|| A ||_{\mathrm{op}} = \sigma_{\mathrm{max}}(A)$, from which the resulting prediction regions take on the form of $\mathcal{B}_{\widehat{q}}(f(\theta))$, namely a ball of radius $\widehat{q}$, the conformal quantile, under the $||\cdot||_{\mathrm{op}}$ metric. 

\subsection{Coverage Guarantee Consequences}\label{section:cov_guar}
We now characterize the regret induced by the robustness across LQR setups, that is $\mathcal{R}(\theta, C) := \mathbb{E}[J(K^{*}_{\mathrm{rob}}(\mathcal{U}(\theta)), C) - J(K^{*}(C), C)]$,
where the randomness is over stochastics in the \textit{true} system dynamics $C := [A,B]$ and in the $\mathcal{P}(C\mid\theta)$ map. We explicitly note $C$ in the regret notation to emphasize that, while the controller $K$ is defined using \textit{estimated} system dynamics, the final evaluation over the \textit{true} $C$ dynamics. 

We provide the regret statements for the continuous, infinite time horizon cases below and defer the discrete-time and finite time horizon cases to the Appendix. Both settings require a mild assumption that the problem parameters have bounded norms, formalized in \Cref{assump:opt_technical_1}; this will hold for any realistic problem setup. Notably, however, the two settings differ in that the stochastic dynamics requires $Q(t)$ and $R(t)$ be discounted over $t$, while the deterministic case is fully compatible with non-discounted rewards. Intuitively, this discounting is necessary, as stability alone in the stochastic setting does not ensure a bounded objective; the state can continue to oscillate and result in an unbounded accumulation of error if the terms tied to the state covariance matrix do not decay. Other works frame this assumption as ``mean-square stability,'' (see e.g. \cite{gravell2020learning}). We formally pose this as \Cref{assump:opt_technical_2}. 

\begin{assumption}\label{assump:opt_technical_1}
    For any $\theta$, $K\in\mathcal{K}(\mathcal{B}_{\widehat{q}}(f(\theta)))$, and $\widehat{C}\in\mathcal{B}_{\widehat{q}}(f(\theta))$, $D(K) := \sqrt{n} || Q+K^\top  R K||_{\infty} || x_{0} ||_{\infty}^{2} || W ||_{\mathrm{op}} < \infty$
\end{assumption}

\begin{assumption}\label{assump:opt_technical_2}
    For any $\theta$, $\exists$ constants $\alpha_{1},\beta_{1} >0$ such that for all $\widehat{C}\in\mathcal{B}_{\widehat{q}}(f(\theta))$, $K\in\mathcal{K}(\widehat{C})$, and $t\ge 0$, $||Q(t) + K^\top R(t) K|| \le \beta_{1} e^{-\alpha_{1} t}$ and $\min_{\widehat{C}\in\mathcal{B}_{\widehat{q}}(f(\theta))} (2 \alpha_{2}(\widehat{C}) + \alpha_{1}) > 0$ where $\alpha_{2}(\widehat{C}) := \max_{K\in\mathcal{K}(\widehat{C})} (-\max_{i} \mathrm{Re}(\lambda_i(\widehat{C}W))) >0$.
\end{assumption}

The regret bound below decomposes into two terms. The first captures the suboptimality in designing a controller with the conformal dynamics set instead of against the true dynamics; this coincides with the suboptimality characterized in previous works described in \Cref{section:bg_pred_opt}. The other is a novel aspect that arises in this controls setting: since the robust control problem optimizes over a \textit{restricted} set of controllers, namely those that universally stabilize the full conformal dynamics set instead of those that only stabilize the true dynamics, there is an additional ``domain gap'' suboptimality. Intuitively, if the true optimal controller falls in $\mathcal{K}(\mathcal{B}_{\widehat{q}}(f(\theta)))$, this latter term should vanish. Towards this end, we introduce the following notion.

\begin{definition}\label{def:margin}
Let $M(C,K^*(C)) := A - B K^*(C)$ be diagonalizable. Define $r(C,K^*(C)) := \frac{\min_i\big(-\Re(\lambda_i(M))\big)}{\kappa(U)\,\|W\|_{\mathrm{op}}}$, where $M = U \Lambda U^{-1}$, $\kappa(U)$ is the condition number of $U$, and $W=[I\;-\!K^*(C)^\top]^\top$.
\end{definition}

Across the theorems stated below, therefore, if the conformal radius is smaller than this $r(C,K^*(C))$ margin term, the ``domain gap'' term vanishes. Intuitively, this property follows as the stability of the system can be characterized by the closed-loop eigenvalues, whose values change by a bounded amount in considering the perturbations captured in the conformal region. We additionally see that, as $\widehat{q}\to0$, the suboptimality vanishes, as we would expect in recovering the true dynamics. Thus, users should seek to produce prediction regions with coverage that are as small as possible to produce informative upper bounds on the nominal optimal value. Notably, the statements below are given in terms of the objective Lipschitz constants $L$: explicit expressions of $L$ along with the discrete-time and finite time horizons theorems and proofs are given in \Cref{section:coverage_bound_control_disc_det,section:coverage_bound_control_cont_det,section:coverage_bound_control_disc_stoch,section:coverage_bound_control_cont_stoch}.

\begin{theorem}[Deterministic, continuous-time]
    Let $J(K, C) := \int_{0}^{\infty} (x(t)^\top (Q + K^\top R K) x(t)) dt$ for $w=0$. Assume that $\mathcal{P}_{\Theta,C}(C\in\mathcal{B}_{\widehat{q}}(f(\Theta))) \ge 1 - \alpha$. Then, under \Cref{assump:opt_technical_1},
    \begin{equation*}
        \mathcal{P}_{\Theta,C}\left(0\le \mathcal{R}(\Theta, C) \le 2L \widehat{q} + \Delta_{\mathrm{dom}}(\Theta,C)\right) \ge 1 - \alpha,
    \end{equation*}
    where $L$ is the Lipschitz constant of $J(K,\widehat{C})$ in $\widehat{C}\in\mathcal{B}_{\widehat{q}}(f(\theta))$ under the operator norm. Further, if $\widehat{q} < r(C,K^*(C))$, (see \Cref{def:margin}) $\Delta_{\mathrm{dom}}(\Theta,C)=0$.
\end{theorem}



\begin{theorem}[Stochastic, continuous-time]
    Let $J(K, C) := \mathbb{E}\left[\int_{0}^{\infty} (x(t)^\top (Q(t) + K^\top R(t) K) x(t)) dt \right]$ with $w(t)$ a white noise process with spectral density $\Sigma$ such that $D_{2}(K) := || \Sigma ||_{\mathrm{op}} || W ||_{\mathrm{op}} < \infty$. Assume further that $\mathcal{P}_{\Theta,C}(C\in\mathcal{B}_{\widehat{q}}(f(\Theta))) \ge 1 - \alpha$. Then, under \Cref{assump:opt_technical_1} and \Cref{assump:opt_technical_2},
    \begin{equation*}
        \mathcal{P}_{\Theta,C}\left(0\le \mathcal{R}(\Theta, C) \le 2L \widehat{q} + \Delta_{\mathrm{dom}}(\Theta,C)\right) \ge 1 - \alpha,
    \end{equation*}
    where $L$ is the Lipschitz constant of $J(K,\widehat{C})$ in $\widehat{C}\in\mathcal{B}_{\widehat{q}}(f(\theta))$ under the operator norm. Further, if $\widehat{q} < r(C,K^*(C))$, (see \Cref{def:margin}) $\Delta_{\mathrm{dom}}(\Theta,C)=0$.
\end{theorem} 


\subsection{Ambiguous Ground Truth}\label{section:ambiguous_ground_truth}
We now discuss the complication of obtaining coverage guarantees on $C$ despite only observing estimates $\widetilde{C} = C + \epsilon$ in the dataset, where $\epsilon\sim\mathcal{N}(0,\Sigma)$. This form of the estimation error can be shown to hold asymptotically under mild assumptions by classical results from least squares estimation, as shown for LTI system identification in \citep{ljung1987theory}.

The coverage guarantee result given in \Cref{thm:amb_ground_truth} is the multivariate extension of Theorem A.5 from \citep{feldman2023cp} and is a novel contribution to the broader space of conformal prediction. Intuitively, we show that if, for all $\theta$, the density $\mathcal{P}(C\mid\theta)$ peaks in $\mathcal{U}(\theta)$, we retain marginal coverage guarantees. If $\mathcal{P}(C\mid\theta)$ is unimodal and radially symmetric about its mode, this condition is satisfied so long as $\mathcal{U}(\theta)$ captures the mode. The map between design parameters $\theta$ and $A,B$ is often unimodal, making such a structural assumption reasonable; this was true classically, where a deterministic map was parametrically given by physics (discussed more in \Cref{section:unimodal_assump}), and remains true of data-driven surrogates in UCCD \citep{azad2023concurrent,azad2024site}. $\mathcal{U}(\theta)$ capturing the mode is also a weak assumption assuming a zero-centered distribution for $\epsilon$, since it then amounts to capturing the mode of $\mathcal{P}(\widetilde{C}\mid\theta)$, which holds for any sufficiently accurate predictor. We empirically demonstrate that such assumptions hold and, thus, that the coverage guarantees are retained in \Cref{section:experiments}. The full proof of this theorem is deferred to \Cref{section:coveragefornoisyobs}. 

\begin{theorem}\label{thm:amb_ground_truth}
    Let $\widetilde{C} = C + \epsilon$ where $\text{vec}(\epsilon) \sim \mathcal{N}(0, \Sigma)$, where $\epsilon\indep (\Theta, C)$. Assume $\mathcal{U}(\theta) = \{C' \mid || f(\theta) - C' ||_{\mathrm{op}} \le \widehat{q}\}$ satisfies $\mathcal{P}_{\Theta,\widetilde{C}}(\widetilde{C} \in \mathcal{U}(\Theta)) \ge 1 - \alpha$, where $|| \cdot ||_{\mathrm{op}}$ denotes the matrix operator norm. If for any $\theta\in\Theta$ and $\delta > 0$, $\mathcal{P}(\widehat{q}^2 - \delta \leq \|C - f(\theta)\|_{\mathrm{op}}^{2}   \leq \widehat{q}^2 \mid\Theta = \theta) > \mathcal{P}(\widehat{q}^2 \leq \|C - f(\theta)\|_{\mathrm{op}}^{2}   \leq \widehat{q}^2 + \delta \mid\Theta = \theta)$, then
    \begin{align*}
        \mathcal{P}_{\Theta, C}(C \in \mathcal{U}(\Theta)) \geq \mathcal{P}_{\Theta,\widetilde{C}}(\widetilde{C} \in \mathcal{U}(\Theta)) \geq 1 - \alpha.
    \end{align*}
\end{theorem}

\subsection{Optimization Algorithm}\label{section:optimization}

Due to our generalization over traditional approaches to robust control (discussed in detail in \Cref{section:related_works}), the standard approaches of solution used in those cases, namely generalized algebraic Riccati equations (GARE) or policy gradient, cannot be applied without modification. We, thus, now discuss how policy gradient can be adapted to efficiently solve the problem of interest and then demonstrate corresponding convergence results in \Cref{section:opt_conv_guar}. Given the novelty of the framing of \Cref{eqn:robust_objective} over previous framings, specifically in the geometry of the uncertainty regions, the policy gradient expressions derived here too are novel, highlighted below.
We frame this discussion around the deterministic, discrete-time, infinite time horizon setting, in which $x_{t+1} = (A - BK) x_{t}$. We assume that the initial state is drawn from a known distribution $x(0)\sim\mathcal{N}(0, X_0)$. 
Naively, computing the gradient would require estimation of the infinite sum in $J$; however, it is well known that the gradient can be computed using a Lyapunov formulation, given by
\begin{equation}\label{eqn:policy_grad}
    \nabla_{K} J(K, A, B) = 2((R + B^{\top}P_{K}B)K - B^{\top} P_{K} A) X_K,
\end{equation}
where $X_K$ and $P_{K}$ respectively solve the two Lyapunov equations $\ell_X(X_K, \Delta_{K}) = 0$ and $\ell_P(P_K, \Delta_{K}, K) = 0$ for specified $Q, R, K$, and $\Delta_{K} := A - BK$, where
\begin{gather}\label{eqn:cts_lyapunov}
    \ell_X(X_K, \Delta_{K}) := \Delta_{K} X_K \Delta^{\top}_{K} - X_K + X_0 \\
    \ell_P(P_K, \Delta_{K}, K) := \Delta^{\top}_{K} P_{K} \Delta_{K} + Q + K^{\top} R K - P_{K} \nonumber 
\end{gather}
Note that, while $\ell_P$ also depends on the choice of $Q$ and $R$, we do not explicitly note this in the notation as they remain fixed throughout the problem. If the continuous-time setting is of interest instead,
there are analogous Lyapunov equations and gradient expressions to those respectively in \Cref{eqn:policy_grad} and \Cref{eqn:cts_lyapunov}. To solve \Cref{eqn:robust_objective}, we wish to perform gradient updates on $K$ instead with respect to $\phi(K) := \max_{\widehat{C} \in \mathcal{U}(\theta)} J(K, \widehat{C})$. Naively, one could proceed through the remaining analysis by leveraging Danskin's Theorem to compute the gradient of $\nabla_{K} \phi(K)$, which would result in the expression $\nabla_{K} \phi(K) = \nabla_{K} J(K, C^{*}(K))$, where $C^{*}(K) := \argmax_{\widehat{C} \in \mathcal{U}(\theta)} J(K, \widehat{C})$; however, the existence of such a gradient requires that $C^{*}(K)$ be the unique maximizer. Such an assumption is unlikely to hold in practice; for this reason, we instead relax this assumption and proceed using subgradients. That is, we suppose $C^{*}(K) := \Argmax_{\widehat{C} \in \mathcal{U}(\theta)} J(K, \widehat{C})$ instead is a \textit{set} and denote by $\partial_{K} \phi(K) := \{\nabla_{K} J(K, C_K) : C_K\in C^{*}(K)\}$ the \textit{vertices} of the subdifferential.


Thus, robust policy optimization proceeds by iteratively updating $K$ with \textit{any} vertex of the subdifferential, namely by evaluating \Cref{eqn:policy_grad} with $[A_K,B_K] := C_K$ for some $C_K\in C^{*}(K)$ and $(X_K^{*},P_{K}^{*})$, the solutions to the Lyapunov equations when $\Delta^{*}_K := A_K - B_KK$. We initialize this procedure with the optimal controller for the nominally predicted dynamics, i.e. $K^{(0)} := K^{*}(f(\theta))$. Extending LQR policy gradient methods to the robust setting, therefore, reduces to being able to efficiently solve the maximization problem of $C_K$ over $\mathcal{U}(\theta) := \mathcal{B}_{\widehat{q}}(f(\theta))$. This can be estimated with gradient ascent, where the Lyapunov expression 
$\nabla_{C} J(K, C) = 2 P_{K} C W X_K W^{\top}$ is derived in \Cref{section:lyap_c_grad}. The use of subgradients for policy optimization and the derivation of the explicit $\nabla_{C} J(K, C)$ expression in its Lyapunov formulation are novel contributions to the robust LQR space; these aspects were heretofore unstudied as previously studied robust formulations (in \Cref{section:related_works}) could be translated into GAREs and, therefore, did not require algorithmic innovation.
The algorithm is given in \Cref{alg:crc}.

\begin{algorithm}
  \caption{\label{alg:crc} \textproc{Conformalized Predict-Then-Control (CPC)}}
  \begin{algorithmic}[1]
    \Procedure{CPC}{$\theta, f(\theta), \widehat{q}, \eta_{K}, \eta_{C}, T_{K}, T_{C}$}
    \Statex \textbf{Inputs: } Design $\theta$, Predictor $f(\theta)$, Conformal quantile $\widehat{q}$, Step sizes $\eta_{K},\eta_{C}$, Max steps $T_{K},T_{C}$
    \State $\widehat{C} := f(\theta), K^{(0)} \gets \textsc{SolveARE}(\widehat{C})$
    \For{$t_{K} \in\{0,\ldots T_{K}-1\}$}
      \State $[A^{(0)},B^{(0)}] := C^{(0)}\gets\widehat{C}$
      \For{$t_{C} \in\{0,\ldots T_{C}-1\}$}
        \State $\Delta^{(t_{C})} := A^{(t_{C})} - B^{(t_{C})} K^{(t_{K})}$
        \State $X^{(t_{C})} \gets \mathrm{Solve}(\ell_{X}(X, \Delta^{(t_{C})}) = 0; X)$
        \State $P^{(t_{C})} \gets \mathrm{Solve}(\ell_{P}(P, \Delta^{(t_{C})}, K^{(t_{K})}) = 0; P)$
        \State $\begin{aligned}[t]
            C&^{(t_{C}+1)}  \gets {}\Pi_{\mathcal{B}_{\widehat{q}}(\widehat{C})}\bigl(
               C^{(t_{C})} +  \\
            & \quad\eta_{C}\,(2 P^{(t_{C})} C^{(t_{C})} W^{(t_{K})} X^{(t_{C})}(W^{(t_{K})})^{\top})\bigr)
            \end{aligned}$
      \EndFor
      \State $\Delta^{*} := A_K - B_K K^{(t_{K})} \qquad\triangleright C^{*} := C^{(T_C)}$
      \State $X^{*} \gets\mathrm{Solve}(\ell_{X}(X, \Delta^{*}) = 0; X)$
      \State $P^{*} \gets \mathrm{Solve}(\ell_{P}(P, \Delta^{*}, K^{(t_{K})}) = 0; P)$
      \State $\begin{aligned}[t]
            K&^{(t_{K}+1)} \gets K^{(t_{K})} - \eta_{K} \bigl(2((R \\
            & + (B_K)^{\top}P^{*}B_K)K^{(t_{K})} - (B_K)^{\top} P^{*} A_K) X^{*}\bigr)
            \end{aligned}$
    \EndFor
    \State \textbf{Return} $K^{(T_K)}$
    \EndProcedure
  \end{algorithmic}
\end{algorithm}

\subsection{Policy Gradient Convergence Guarantees}\label{section:opt_conv_guar}
We now wish to demonstrate this policy gradient approach retains the desired convergence properties it satisfies in the nominal case. Convergence guarantees surprisingly hold in the standard case despite the nonconvexity of the problem in $K$ due to a property known as ``gradient dominance'' \citep{gravell2020learning}. A function $f : \mathbb{R}^{d_{1}\times d_{2}}\rightarrow\mathbb{R}$ is gradient-dominated if, for some $\mu>0$, $f(x) - f(x^{*})\le\mu|| \nabla_{x} f(x) ||_{F}^{2}$, where $x^{*} := \argmin_{x} f(x)$. 

We proceed through the analysis similarly leveraging gradient dominance; however, our analysis has the novel problem of having to handle the non-uniqueness of the subgradient being used, namely that our algorithm may perform updates with one of the collection of subdifferential vertices than using the uniquely defined gradient. For this reason, we instead consider a generalized notion of \textit{subgradient} domination, defined as $\exists$ some $\mu>0$ such that $f(x) - f(x^{*})\le\mu \min_{g\in\partial f(x)} || g ||_{F}^{2}$. We show that $\phi$ satisfies subgradient dominance and that this then produces convergence guarantees for \Cref{alg:crc} in \Cref{lemma:grad_dom}. 

The full proof is deferred to \Cref{section:pg_conv_guar} and parallels the proof strategy presented in \citep{fazel2018global}; the main technical challenges are in demonstrating that bounds on expressions related to $J(K, C)$ and $\nabla_{K} J(K, C)$ are retained in our robust setting and that the non-uniqueness of the maximizer does not interfere with convergence. Intuitively, the non-uniqueness manifests as a looser gradient dominance constant and, thus, convergence decay rate, since $\mu$ must be taken to be the loosest constant amongst those of the maximizing set. In line with \citep{fazel2018global}, we assume $X_{K}\succcurlyeq 0$ across $\widehat{C}\in\mathcal{C}$ and $K\in\mathcal{K}(\mathcal{C})$. This is true if the system is controllable for any $\widehat{C}\in\mathcal{C}$, which holds if the nominal dynamics are well-behaved and the predictor $f(\theta)$ is sufficiently accurate, resulting in a small $\mathcal{C}$ set. The statements below are made for a general set of dynamics $\mathcal{C}$, though we are interested in $\mathcal{C} := \mathcal{U}(\theta)$. We defer the presentation of the explicit poly-expression in \Cref{thm:pol_grad_convergence} to \Cref{section:pg_conv_guar}.


\begin{theorem}\label{thm:pol_grad_convergence}
    Let $J(K, C) := \sum_{t=0}^{\infty} (x_{t}^\top (Q + K^\top R K) x_{t})$ for $w=0$. Let $K^{(t)}$, $\phi(K) := \max_{C \in \mathcal{C}} J(K, C)$, and $K^{*}_{\mathrm{rob}}(\mathcal{C}) := \argmin_{K\in\mathcal{K}(\mathcal{C})} \phi(K)$ be the $t$-th iterate of \Cref{alg:crc}. Assume for each iterate $t$, the optimization over $C$ converges, i.e. $C^{(T_C)} = C^{*}(K^{(t)})$, that $K^{(t)}\in\mathcal{K}(\mathcal{C})$, and that $X_{K}\succcurlyeq0$ for all $\widehat{C}\in\mathcal{C}$ and $K\in\mathcal{K}(\mathcal{C})$. Denote $\nu := \min_{\widehat{C}\in\mathcal{C}} \min_{K\in\mathcal{K}(\mathcal{C})} \sigma_{\min}(X_{K})$. If in \Cref{alg:crc} 
        $\eta_{K}\le\min_{[\widehat{A},\widehat{B}] \in \mathcal{C}} \mathrm{poly}(\frac{\nu\sigma_{\min}(Q)}{J(K^{(0)}, C)}, \frac{1}{|| \widehat{A} ||}, \frac{1}{|| \widehat{B} ||}, \frac{1}{|| R ||}, \sigma_{\min}(R)),$
    then, there exists a $\gamma > 0$ such that $\phi(K^{(T)}) - \phi(K^{*}_{\mathrm{rob}}(\mathcal{C}))
    \le (1 - \gamma)^{T} (\phi(K_{0}) - \phi(K^{*}_{\mathrm{rob}}(\mathcal{C}))).$
\end{theorem}

Formally, such convergence is guaranteed only if iterates $K^{(t)}$
remain within $\mathcal{K}(\mathcal{U}(\theta))$. One modification to \Cref{alg:crc} would involve projecting intermediate iterates to this stabilizing set by solving
\begin{equation}\label{eqn:projection_stab}
\begin{gathered}
    \Pi_{\mathcal{K}(\mathcal{U}(\theta))}(\widetilde{K}^{(t)}) := 
    \argmin_{K} ||K - \widetilde{K}^{(t)}||_{\mathrm{op}} \\
    \textrm{\ \ s.t.\ \ } \max_{[\widehat{A},\widehat{B}] := \widehat{C}\in\mathcal{U}(\theta)} \max_{i} \mathrm{Re}(\lambda_i(\widehat{A} + \widehat{B}K)) < 0 \nonumber.
\end{gathered}
\end{equation}
There, however, is no known efficient algorithm to solve this projection step. Despite being of theoretical concern, this instability issue fails to be practically relevant, since the controller iterates remain well within the set of stabilization for sufficiently accurate predictors $f(\theta)$. 
If instabilities arise, an approximate solution can be obtained by replacing $\max_{\widehat{C}\in\mathcal{U}(\theta)}$ of \Cref{eqn:projection_stab} with a finite sampling $\{\widehat{C}^{(i)}\}$ over $\mathcal{U}(\theta)$. 

\section{Related Works}\label{section:related_works}
Robust control can be broadly categorized into trajectory-based and trajectory-free robustness. The former adjusts an initially posited control scheme in an \textit{online} fashion based on feedback measurements \citep{azad2022control,seiler2020introduction,paraskevopoulos2017modern}, whereas the latter directly incorporates desired robustness into the optimization problem \textit{prior} to deployment \citep{gravell2020robust,gravell2022robust}. 
Given that control co-design seeks to identify a controller \textit{prior} to deploying a design, we specifically highlight methods of trajectory-free robust control.

A popular classical trajectory-free method is $\mathcal{H}_{\infty}$ control. $\mathcal{H}_{\infty}$ is typically formulated as minimizing $|| T_{wz} ||_{\infty}$, i.e. the frequency space transfer function from $w\to z$ for some performance state $z$. By defining $z = \begin{bmatrix}
    Q^{1/2}x & R^{1/2}u 
\end{bmatrix}$, we recover a recognizable LQR formulation, with the objective replaced by
\begin{equation}\label{eqn:h_inf}
    u^{*} = \min_{\{u_t\}} \max_{\{w_t\}} \sum_{t=0}^{\infty} (x_t^{\top} Q x_t + u_t^{\top} R u_t - \gamma^2 w_t^{\top} w_t),
\end{equation}
which can be solved via generalized Riccati equations \citep{bacsar2008h}. Here, $\gamma$ can either be fixed to perform suboptimal $\mathcal{H}_{\infty}$ synthesis or it can be determined via bisection to identify the smallest $\gamma$ such that a solution exists. Notably, the nominal $\mathcal{H}_{\infty}$ formulation seeks additive, unstructured disturbance rejection. Of interest herein, however, was robustness to \textit{multiplicative} uncertainties through the system dynamics. Towards this end, $\mu$-synthesis offers an extension to $\mathcal{H}_{\infty}$ control by allowing users to specify norm-bounded uncertainties on system dynamics \citep{bevrani2015robust,chen2014mu}.

This need for manual specification in $\mu$-synthesis, however, incurs conservatism or controller instability if poorly specified, resulting in increasing interest in data-driven specifications. In this vein, a formulation known as LQR with multiplicative noise (LQRm), has recently become of interest, where the controller is:
\begin{gather}\label{eqn:lqrm}
K^* := \argmin_{K} \mathbb{E}_{\{\delta_i\},\{\gamma_i\}}[J(K, A, B)] \\
\dot{x} := \left(A + \sum_{i=1}^{p} \delta_{i} A_i\right)x + \left(B + \sum_{i=1}^{q} \gamma_{i} B_i\right)u + w, \nonumber
\end{gather}
where $\{A_i\}$ and $\{B_i\}$ and the distributions of $\{\delta_i\}\sim\mathcal{D}_{\delta}$ and $\{\gamma_i\}\sim\mathcal{D}_{\gamma}$ can be specified, either with data-free or with data-driven estimation. Most common among data-free specifications are so-called ``margin methods.'' Briefly, margin methods specify $\{\delta_i\}$ and $\{\gamma_i\}$ by finding those $\{\delta_i\}$ and $\{\gamma_i\}$ that result in borderline-stable dynamics when paired with the corresponding, manually specified $(\{A_i\},\{B_i\})$ and some choice of controller: the particular controller varies across margin strategies. A full description of the margin methods considered is given in \Cref{section:exp_setup}. 

As with $\mathcal{H}_{\infty}$, such data-free LQRm methods sacrifice stability or risk conservatism in ignoring the nature of the dynamics predictor misspecification, resulting in recent works that give data-driven parameterizations \citep{gravell2020robustdesign}. Here, two approaches were proposed to learn the margin parameters, which we refer to as the ``Shared Lyapunov'' and ``Auxiliary Stabilizer'' approaches, described fully in their paper. While such approaches improve upon the conservatism of classical data-free margin methods, they still require the hand specification of the perturbation matrices $\{A_i\}$ and $\{B_i\}$.


\section{Experiments}\label{section:experiments} 
\begin{table*}[t]
\caption{\label{table:p_values} Each of the results below are the p-values of paired t-tests conducted pairwise between methods testing $H_{1} : \mathcal{R}^{(\mathrm{CPC})}_{\%} < \mathcal{R}^{(\mathrm{alt})}_{\%}$ over 1,000 i.i.d. test samples. For any comparison method with $>80\%$ unstable cases (see \Cref{table:stab_perc} for percentages), we have marked the entry with ``---''.}
\centering
\resizebox{\textwidth}{!}{%
\begin{tabular}{lccccc}
\toprule
 & Airfoil & Load Positioning & Furuta Pendulum & DC Microgrids & Fusion Plant \\
\midrule
Random Critical & --- & --- & --- & --- & --- \\
Random OL MSS (Weak) & \textbf{0.0003} & --- & --- & --- & --- \\
Random OL MSUS & --- & --- & --- & --- & --- \\
Row-Col Critical & --- & --- & --- & --- & --- \\
Row-Col OL MSS (Weak) & \textbf{0.0117} & --- & --- & --- & --- \\
Row-Col OL MSUS & \textbf{0.0009} & --- & --- & --- & --- \\
Shared Lyapunov & \textbf{0.0001} & \textbf{0.0000} & \textbf{0.0112} & 0.0913 & \textbf{0.0023} \\
Auxiliary Stabilizer & \textbf{0.0001} & \textbf{0.0005} & \textbf{0.0055} & \textbf{0.0428} & 0.0630 \\
$\mathcal{H}_{\infty}$ & \textbf{0.0009} & \textbf{0.0000} & \textbf{0.0071} & \textbf{0.0004} & \textbf{0.0013} \\
\bottomrule
\end{tabular}
}
\end{table*}

We now study five setups of interest in the infinite horizon, discrete-time, deterministic setting, namely LQR control of an airfoil \citep{chrif2014aircraft}, a load positioning system \citep{ahmadi2023lqr,jiang2016iterative}, a Furuta pendulum \citep{arulmozhi2022kalman}, a DC microgrid \citep{liu2023novel}, and a nuclear plant \citep{kirgni2023lqr}. The dimensions of $(\theta,A,B)$ are $(\mathbb{R}^{15},\mathbb{R}^{4\times 4},\mathbb{R}^{4\times 2})$ for the airfoil, $(\mathbb{R}^{5},\mathbb{R}^{4\times 4},\mathbb{R}^{4\times 1})$ for the load positioning system, $(\mathbb{R}^{9},\mathbb{R}^{4\times 4},\mathbb{R}^{4\times 1})$ for the Furuta pendulum, $(\mathbb{R}^{17},\mathbb{R}^{9\times 9},\mathbb{R}^{9\times 1})$ for the DC microgrid, and $(\mathbb{R}^{26},\mathbb{R}^{8\times 8},\mathbb{R}^{8\times 1})$ for the nuclear plant. The full setup details are provided in \Cref{section:exp_system_descriptions}.

We compare against $\mathcal{H}_{\infty}$ control with $\gamma$ bisection, the data-free margin methods, and the data-based methods for the LQRm setup as discussed in \Cref{section:related_works}. The data-free margin methods are as implemented by \citep{gravell2020robust} and are fully described in \Cref{section:exp_setup}, of which we specifically consider ``Random Critical,'' ``Random OL MSS (Weak),'' ``Random OL MSUS,'' ``Row-Col Critical,'' ``Row-Col OL MSS (Weak),'' and ``Row-Col OL MSUS.'' The data-based methods are the ``Shared Lyapunov'' and ``Auxiliary Stabilizer'' approaches from \citep{gravell2020robustdesign}. As discussed, we are considering the trajectory-\textit{free} setting, so we do not compare against methods that achieve robustness adaptively over trajectories, such as those in \citep{gravell2020robust,gravell2022robust}. 


In the experiments, we construct $\mathcal{D}$ as per \Cref{section:problem_formulation}, using random gain matrices $K^{(i)}$.
$N$ was taken to be $2,000$ with $|\mathcal{D}_{\mathcal{C}}| = 400$ and the remaining $\mathcal{D}_{T}$ used to train $f(\theta)$, taken to be feed-forward neural networks. 

\subsection{Robust Control Regret \& Stability}\label{section:experiments_regret}

We first study the empirical regret across the aforementioned systems and robust control methods over $1,000$ i.i.d. test points from $\mathcal{P}(\Theta, C)$. To make results comparable across $\theta^{(i)}$, we normalize each trial by its nominal objective, i.e., $\mathcal{R}_{\%} = \mathcal{R}(\Theta, C) / J(K^{*}(C), C)$, as in \citep{sun2023predict}. If the uncertainty regions of the robust problems are poorly specified, i.e. if the regions of robustness do not capture the true dynamics, the resulting robust controller may have unbounded cost, i.e. $J(K^{*}_{\mathrm{rob}}, C) = \infty$. We, thus, only compute $\mathcal{R}_{\%}$ over the stabilizing controllers and separately report the proportion of destabilized cases. Lower values are desirable for both. 

For each comparison method, we report the result of a one-side paired t-test of $H_{1} : \mathcal{R}^{(\mathrm{CPC})}_{\%} < \mathcal{R}^{(\mathrm{alt})}_{\%}$ in \Cref{table:p_values}. We defer the presentation of the raw regret values and the percent of cases with stabilized dynamics to \Cref{section:addn_exp_results} due to space constraints. Notably, the alternative approaches generally incur greater regret than CPC. For $\mathcal{H}_{\infty}$, this is expected as the misspecification here is in the dynamics matrices, differing from the adversarial exogenous noise that $\mathcal{H}_{\infty}$ is designed to protect against. Similarly, the data-free margin methods protect against perturbations that are misaligned with the true dynamics misspecification, which result in significant instability for the higher-dimensional problems (i.e. the ``Furuta Pendulum,'' ``DC Microgrids,'' and ``Fusion Plant'' tasks). The data-driven LQRm methods improve significantly upon these margin approaches in stability, yet they are too conservative as they do not make use of the anticipated structures of the errors made in the predictions by $\widehat{f}(\theta)$.

\begin{figure}
  \centering
  \includegraphics[width=.5\columnwidth]{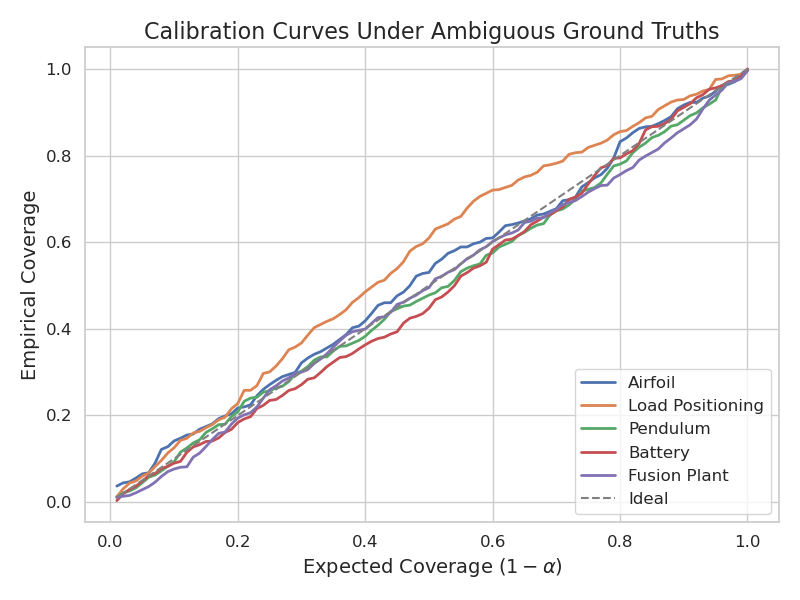}
  \caption{Calibration plots for the tasks, assessed on 1,000 i.i.d. test samples of $C$ with calibration performed using the estimated $\widetilde{C}$, affirming \Cref{thm:amb_ground_truth}.}
  \label{fig:cal}
\end{figure}


\subsection{Ambiguous Ground Truth Calibration}

To validate the results of \Cref{thm:amb_ground_truth} and demonstrate the empirical validity of the associated assumption, we computed the empirical coverages across various levels of desired coverage $\alpha\in(0,1)$ for the experimental setups. As previously discussed, the calibration here was performed using a calibration set of \textit{estimated} dynamics $\widetilde{\mathcal{D}} = \{(\theta^{(i)}, \widetilde{C}^{(i)})\}$ but coverage was assessed on the \textit{true} dynamics $\{C^{(i)}\}$. We computed this in the manner described in \Cref{section:experiments} for $\alpha$ varying by increments of 0.05. For assessing coverage, we again used $1,000$ test points drawn i.i.d. from $\mathcal{P}(\Theta, C)$ and measured the proportion of samples for which $s(\theta^{(i)},C^{(i)})\le\widehat{q}$. The results are shown in \Cref{fig:cal}, where we see the desired calibration under calibration with estimated dynamics.  

\section{Discussion}\label{section:discussion}
We have presented CPC, a principled framework for specifying the LQR robust control subproblem in a UCCD setting, suggesting many directions for extension. 
The most immediate would involve integrating this framework fully into a UCCD pipeline: we focused herein on the robust control subproblem but characterizing the end-to-end workflow is of great interest. 
In addition, nonlinear extension by leveraging Koopman operator theory or nonparametric neural operator models would be interesting \citep{brunton2021modern,mauroy2020koopman,qian2020lift} as would the extension to MDPs \citep{wang2021learning}. 


\printbibliography

\clearpage
\appendix
\thispagestyle{empty}

\onecolumn

\section{Prediction Region Validity Lemma}\label{section:coverage_bound}
Given that we characterize \textit{both} the continuous- and discrete-time settings below, we produce a generalized definition to that presented in \Cref{def:margin}.

\begin{definition}\label{def:margin_extended}
Let $M(C,K^*(C)) := A - B K^*(C)$ be the optimal closed-loop matrix. Define
\begin{equation*}
    r(C,K^*(C)) :=
        \begin{cases}
        \frac{\min_i -\mathrm{Re}(\lambda_i(M))}{\kappa(U)\,||W||_2} & \text{Continuous time setting}\\
        \frac{\min_i (1-|\lambda_i(M)|)}{\kappa(U)\,||W||_2} & \text{Discrete time setting}
        \end{cases}
\end{equation*}
where $M = U \Lambda U^{-1}$, $\kappa(U)$ is the condition number of $U$, and $W=[I\;-\!K^*(C)^\top]^\top$.
\end{definition}

\begin{lemma}\label{lemma:coverage_bound}
    Let $J(K, C)$ be a function such that, for any fixed $\theta$, it is non-negative and $L$-Lipschitz in $\widehat{C}\in\mathcal{B}_{\widehat{q}}(f(\theta))$ under the operator norm for any $K\in\mathcal{K}(\mathcal{U}(\theta))$,
    where $\mathcal{K} : \Omega(\Theta)\rightarrow\Omega(\mathcal{C})$ and $\Omega_{1}\subset\Omega_{2}\implies\mathcal{K}(\Omega_{2})\subset \mathcal{K}(\Omega_{1})$. Further assume that $\mathcal{P}_{\Theta,C}(C \in \mathcal{U}(\Theta)) \ge 1 - \alpha$. Then:
    \begin{equation}
        \mathcal{P}_{\Theta,C}\left(0\le \mathcal{R}(\Theta, C) \le 2L \widehat{q} + \Delta_{\mathrm{dom}}(\Theta,C)\right) \ge 1 - \alpha.
    \end{equation}
    Further, if $\widehat{q} < r(C,K^*(C))$, (see \Cref{def:margin_extended}) $\Delta_{\mathrm{dom}}(\Theta,C)=0$.
\end{lemma}

\begin{proof}
    We consider the event of interest conditionally on a pair $(\theta, C)$ where $\widehat{C}\in\mathcal{B}_{\widehat{q}}(f(\theta))$. By assumption, we then have that $\mathcal{K}(\mathcal{B}_{\widehat{q}}(f(\theta)))\subset\mathcal{K}(C)$. As previously noted, the suboptimality here is defined over the \textit{true} $C$ matrix, meaning, unlike previous works, we here wish to bound $J(K^{*}(\mathcal{B}_{\widehat{q}}(f(\theta))), C) - \min_{K\in\mathcal{K}(C)} J(K, C)$ in place of $ \min_{K\in\mathcal{K}(C)} \max_{\widehat{C}\in\mathcal{B}_{\widehat{q}}(f(\theta))} J(K, \widehat{C}) - \min_{K\in\mathcal{K}(C)} J(K, C)$, where $K^{*}(\mathcal{B}_{\widehat{q}}(f(\theta))) := \argmin_{K\in\mathcal{K}(\mathcal{B}_{\widehat{q}}(f(\theta)))} \max_{\widehat{C}\in\mathcal{B}_{\widehat{q}}(f(\theta))} J(K, \widehat{C})$. We begin by matching the minimization sets in the terms as follows:
    \begin{gather*}
        \left|J(K^{*}(\mathcal{B}_{\widehat{q}}(f(\theta))), C) - \min_{K\in\mathcal{K}(C)} J(K, C)\right| \\
        = \left|J(K^{*}(\mathcal{B}_{\widehat{q}}(f(\theta))), C) - \min_{K\in\mathcal{K}(\mathcal{B}_{\widehat{q}}(f(\theta)))} J(K, C) + \min_{K\in\mathcal{K}(\mathcal{B}_{\widehat{q}}(f(\theta)))} J(K, C) - \min_{K\in\mathcal{K}(C)} J(K, C)\right| \\
        \le \underbrace{\left|J(K^{*}(\mathcal{B}_{\widehat{q}}(f(\theta))), C) - \min_{K\in\mathcal{K}(\mathcal{B}_{\widehat{q}}(f(\theta)))} J(K, C)\right|}_{\text{statistical robustness cost}} + \underbrace{\left|\min_{K\in\mathcal{K}(\mathcal{B}_{\widehat{q}}(f(\theta)))} J(K, C) - \min_{K\in\mathcal{K}(C)} J(K, C) \right|}_{\text{universal stabilization cost}}
    \end{gather*}
    As discussed in the main text, the error decomposes into two terms: the first from making the controller robust to adversarial dynamics matrices and the second from requiring that such a controller stabilize the whole collection of dynamics. We now bound each term separately, starting with the first term. We first note:
    \begin{gather*}
        J(K^{*}(\mathcal{B}_{\widehat{q}}(f(\theta))), C)
        \le \max_{\widehat{C}\in\mathcal{B}_{\widehat{q}}(f(\theta))} J(K^{*}(\mathcal{B}_{\widehat{q}}(f(\theta))), \widehat{C})
        =: \min_{K\in\mathcal{K}(\mathcal{B}_{\widehat{q}}(f(\theta)))} \max_{\widehat{C}\in\mathcal{B}_{\widehat{q}}(f(\theta))} J(K, \widehat{C})
    \end{gather*}
    where the first step follows by the assumption $C\in\mathcal{B}_{\widehat{q}}(f(\theta))$ and second by definition of $K^{*}(\mathcal{B}_{\widehat{q}}(f(\theta)))$. From here,
    \begin{align*}
        \left|J(K^{*}(\mathcal{B}_{\widehat{q}}(f(\theta))), C) - \min_{K\in\mathcal{K}(\mathcal{B}_{\widehat{q}}(f(\theta)))} J(K, C) \right| 
        &\le 
        \left| \min_{K\in\mathcal{K}(\mathcal{B}_{\widehat{q}}(f(\theta)))} \max_{\widehat{C}\in\mathcal{B}_{\widehat{q}}(f(\theta))} J(K, \widehat{C}) - \min_{K\in\mathcal{K}(\mathcal{B}_{\widehat{q}}(f(\theta)))} J(K, C) \right|  \\
        &\le \max_{K\in\mathcal{K}(\mathcal{B}_{\widehat{q}}(f(\theta)))} | \max_{\widehat{C} \in \mathcal{B}_{\widehat{q}}(f(\theta))} J(K, \widehat{C}) - J(K, C) | \\ 
        &\le L \max_{\widehat{C} \in \mathcal{B}_{\widehat{q}}(f(\theta))} || \widehat{C} - C ||_{\mathrm{op}} 
        \le 2L \widehat{q}.
    \end{align*}
    We now demonstrate that the second term vanishes within a radius of ``safety,'' which we do through perturbation analysis of the closed-loop margin. In particular, notice that, since $\mathcal{K}(\mathcal{B}_{\widehat{q}}(f(\theta)))\subset\mathcal{K}(C)$, this term vanishes if the minimizer over $\mathcal{K}(C)$ lies in $\mathcal{K}(\mathcal{B}_{\widehat{q}}(f(\theta)))$. That is, this difference vanishes if $K^{*}(C) := \argmin_{K\in\mathcal{K}(C)} J(K, C)$ satisfies $K^{*}(C)\in\mathcal{K}(\mathcal{B}_{\widehat{q}}(f(\theta)))$. Notably, this is equivalent to finding a condition under which $K^{*}(C)$ stabilizes all the dynamics matrices $\widehat{C}\in\mathcal{B}_{\widehat{q}}(f(\theta))$. 
    
    To procure a test of this property, we consider an approach from the theory of switched linear systems, where controllers are sought that stabilize a collection of adjusting dynamics matrices. In particular, recall that $\mathcal{B}_{\widehat{q}}(f(\theta))$ is constructed under an operator norm and, therefore, any $\widehat{C}\in\mathcal{B}_{\widehat{q}}(f(\theta))$ can be viewed as a bounded perturbation to the nominal prediction corresponding to $\theta$. That is, that $\widehat{C} = C + \Delta$ for $|| \Delta ||_{\mathrm{op}}\le \widehat{q}$. Thus, we can equivalently view the task as seeking a condition that guarantees that, if $|| C - \widehat{C} ||_{\mathrm{op}}\le\widehat{q}$, $K^{*}(C)$ stabilizes $\widehat{C}$.

    We now consider the $W := [I_{n\times n} \mathrm{\ -}K^{*}(C)^{\top}]^{\top}$ matrix and analyze the closed-loop stability of $\widehat{C}W$. By definition, we have that $\widehat{C}W := (C + \Delta)  W =: CW + E$, where we know that $CW$ is stabilized by $K^{*}(C)$, since $K^{*}(C)\in\mathcal{K}(C)$. By this latter point, we know $CW$ satisfies one of two properties, depending on whether the system being analyzed is a continuous- or discrete-time setting. In the case of continuous time, we have that $\min_i -\mathrm{Re}(\lambda_i(CW)) > 0$ and that, for discrete time, $\min_i (1-|\lambda_i(CW)|) > 0$. 

    Under the additional assumption of $CW$ being diagonalizable, we have that $CW = U \Lambda U^{-1}$ for $\Lambda = \mathrm{diag}(\lambda_1, ..., \lambda_n)$. With this, we now return to analyzing the perturbed $CW + E$. By the Bauer–Fike bound, we know that for any eigenvalue $\lambda_j'$ of $CW + E$, $\min_{i} | \lambda_j' - \lambda_i | \le \kappa(U)|| E ||_{\mathrm{op}}$, where $\kappa(U)$ is the condition number of $U$. Thus, if these perturbed eigenvalues remain within the stabilized regions, $\widehat{C}$ is stabilized by $K^{*}(C)$. 

    In the continuous-time setting, this is guaranteed if $\kappa(U)|| E ||_{\mathrm{op}} < \min_i -\mathrm{Re}(\lambda_i(CW))$ or, by the fact that $|| E || := || \Delta W || \le || \Delta || \cdot || W || \le \widehat{q} || W ||$, if
    \begin{equation*}
        \kappa(U) (\widehat{q} || W ||) < \min_i -\mathrm{Re}(\lambda_i(CW)) 
        \iff
        \widehat{q} < \frac{\min_i -\mathrm{Re}(\lambda_i(CW)) }{\kappa(U) || W || }
    \end{equation*}
    The discrete-time setting follows analogously, simply with the stability condition replaced by the discrete-time analog, i.e. if $\kappa(U)|| E ||_{\mathrm{op}} < \min_i (1-|\lambda_i(CW)|)$. That is, a sufficient condition for stabilization is that
    \begin{equation*}
        \kappa(U) (\widehat{q} || W ||) < \min_i (1-|\lambda_i(CW)|)
        \iff
        \widehat{q} < \frac{\min_i (1-|\lambda_i(CW)|)}{\kappa(U) || W || }
    \end{equation*}
    Since we have the assumption that $\mathcal{P}_{\Theta,C}(C \in \mathcal{U}(\Theta)) \ge 1 - \alpha$, the result immediately follows.
\end{proof}

\section{Deterministic Discrete-Time Regret Analysis}\label{section:coverage_bound_control_disc_det}
\begin{theorem}\label{lemma:coverage_bound_control_disc_inf}[Deterministic, discrete-time]
    Let $J(K, C) := \sum_{t=0}^{\infty} (x_t^\top (Q + K^\top R K) x_t)$ with $w=0$. Assume that $\mathcal{P}_{\Theta,C}(C\in\mathcal{B}_{\widehat{q}}(f(\Theta))) \ge 1 - \alpha$. Then, under \Cref{assump:opt_technical_1},
    \begin{equation*}
        \mathcal{P}_{\Theta,C}\left(0\le \mathcal{R}(\Theta, C) \le 2L \widehat{q} + \Delta_{\mathrm{dom}}(\Theta,C)\right) \ge 1 - \alpha,
    \end{equation*}
    where $L$ is the Lipschitz constant of $J(K,\widehat{C})$ in $\widehat{C}\in\mathcal{B}_{\widehat{q}}(f(\theta))$ under the operator norm. Further, if $\widehat{q} < r(C,K^*(C))$, (see discrete-time in \Cref{def:margin_extended}) $\Delta_{\mathrm{dom}}(\Theta,C)=0$.
\end{theorem}
\begin{proof}
We consider any fixed $\theta$ and demonstrate that $J(K, \widehat{C})$ is non-negative and $L$-Lipschitz in $\widehat{C}\in\mathcal{B}_{\widehat{q}}(f(\theta))$ under the operator norm for any $K\in\mathcal{K}(\mathcal{B}_{\widehat{q}}(f(\theta)))$, from which \Cref{lemma:coverage_bound} can be invoked to arrive at the desired conclusion. Given the assumed determinism of the dynamics, we have that $x_{t} = (\widehat{C}W)^{t}x_0$, meaning the above objective setup can equivalently be expressed as:
\begin{equation}
    \sum_{t=0}^{\infty} x_0^\top ((\widehat{C}W)^{t\top} (Q + K^\top R K) (\widehat{C}W)^{t})x_0, 
\end{equation} 
$J$ is clearly non-negative by construction. It, therefore, suffices to demonstrate this objective is Lipschitz continuous with an appropriate Lipschitz constant. Notice the Lipschitz constant can be obtained by bounding the magnitude of the gradient with respect to $\widehat{C}$, which we do as follows
\begin{align*}
    \nabla_{\widehat{C}} \left(\sum_{t=0}^{\infty} x_0^\top ((\widehat{C}W)^{t\top} (Q + K^\top R K) (\widehat{C}W)^{t})x_0 \right)
    &= \sum_{t=0}^{\infty} t \mathrm{diag}((Q+K^\top  R K) (\widehat{C} W)^{t} x_{0}) (\widehat{C} W)^{(t-1)} \mathrm{diag}(x_{0}) W^\top \\ 
    &+ \sum_{t=0}^{\infty} t \mathrm{diag}((Q^\top +(R K)^\top  K) (\widehat{C} W)^{t} x_{0}) (\widehat{C} W)^{(t-1)} \mathrm{diag}(x_{0}) W^\top 
\end{align*}
We now bound the magnitude of this quantity as follows:
\begin{align*}
    L 
    \le \max_{\widehat{C}\in\mathcal{B}_{\widehat{q}}(f(\theta))} || \sum_{t=0}^{\infty} t \mathrm{diag}((Q+K^\top  R K) (\widehat{C} W)^{t} x_{0}) (\widehat{C} W)^{(t-1)} \mathrm{diag}(x_{0}) W^\top \\ 
    + \sum_{t=0}^{\infty} t \mathrm{diag}((Q^\top +(R K)^\top  K) (\widehat{C} W)^{t} x_{0}) (\widehat{C} W)^{(t-1)} \mathrm{diag}(x_{0}) W^\top ||_{\mathrm{op}} \\
    \le \max_{\widehat{C}\in\mathcal{B}_{\widehat{q}}(f(\theta))} \sum_{t=0}^{\infty} t \left|\left| \mathrm{diag}((Q+K^\top  R K) (\widehat{C} W)^{t} x_{0}) (\widehat{C} W)^{(t-1)} \mathrm{diag}(x_{0}) W^\top \right|\right|_{\mathrm{op}} \\ 
    + \sum_{t=0}^{\infty} t \left|\left| \mathrm{diag}((Q^\top +(R K)^\top  K) (\widehat{C} W)^{t} x_{0}) (\widehat{C} W)^{(t-1)} \mathrm{diag}(x_{0}) W^\top \right|\right|_{\mathrm{op}},
\end{align*}
where we have used $\mathrm{diag}(x_{0})$ for a vector $x_{0}$ to denote a diagonal matrix with $x_{0}$ placed along its main diagonal. We now bound each of these two terms separately, although the structure of the two is the same, so we explicitly show steps for bounding the first, from which the same can be repeated on the second. Importantly, we make use of the fact $||\text{diag}(x_{0})||_{\mathrm{op}} = ||x_{0}||_{\infty}$ and $|| A ||_{\infty} \le \sqrt{n} || A ||_{\mathrm{op}}$ for $A\in\mathbb{R}^{n\times n}$ as follows:
\begin{align*}
    \max_{\widehat{C}\in\mathcal{B}_{\widehat{q}}(f(\theta))} &\sum_{t=0}^{\infty} t || \mathrm{diag}((Q+K^\top  R K) (\widehat{C} W)^{t} x_{0}) (\widehat{C} W)^{(t-1)} \mathrm{diag}(x_{0}) W^\top ||_{\mathrm{op}} \\
    &\le \max_{\widehat{C}\in\mathcal{B}_{\widehat{q}}(f(\theta))} \sum_{t=0}^{\infty} t || \mathrm{diag}((Q+K^\top  R K) (\widehat{C} W)^{t} x_{0}) ||_{\mathrm{op}} || (\widehat{C} W)^{(t-1)} ||_{\mathrm{op}} || \mathrm{diag}(x_{0})||_{\mathrm{op}} || W^\top ||_{\mathrm{op}} \\
    &= \max_{\widehat{C}\in\mathcal{B}_{\widehat{q}}(f(\theta))} \sum_{t=0}^{\infty} t || (Q+K^\top  R K) (\widehat{C} W)^{t} x_{0}||_{\infty} ||x_{0}||_{\infty} || (\widehat{C} W)^{(t-1)} ||_{\mathrm{op}} || W ||_{\mathrm{op}} \\
    &\le \max_{\widehat{C}\in\mathcal{B}_{\widehat{q}}(f(\theta))} \sum_{t=0}^{\infty} t || Q+K^\top  R K||_{\infty} || (\widehat{C} W)^{t}||_{\infty} || x_{0}||_{\infty} ||x_{0}||_{\infty} || (\widehat{C} W)^{(t-1)} ||_{\mathrm{op}} || W ||_{\mathrm{op}} \\
    &\le \max_{\widehat{C}\in\mathcal{B}_{\widehat{q}}(f(\theta))} \sum_{t=0}^{\infty} t (\sqrt{n} || Q+K^\top  R K||_{\infty} || x_{0}||^{2}_{\infty} || W ||_{\mathrm{op}}) || (\widehat{C} W)^{t}||_{\mathrm{op}} || (\widehat{C} W)^{(t-1)} ||_{\mathrm{op}}.
\end{align*}
We now collect all terms independent of $t$ into $D(K) = \sqrt{n} || Q+K^\top  R K||_{\infty} || x_{0}||_{\infty}^{2} || W ||_{\mathrm{op}}$:
\begin{align*}
    &\le D(K) \max_{\widehat{C}\in\mathcal{B}_{\widehat{q}}(f(\theta))} \sum_{t=0}^{\infty} t || (\widehat{C} W)^{t}||_{\mathrm{op}} || (\widehat{C} W)^{(t-1)} ||_{\mathrm{op}}.
\end{align*}
Critically, we can now demonstrate that this sum is bounded by virtue of $K$ being a universal stabilizer of the dynamics set $\mathcal{B}_{\widehat{q}}(f(\theta))$. By this stabilization, we know that $\widehat{C} W$ is Schur stable, i.e. $\min_i (1-|\lambda_i(\widehat{C}W)|) > 0$ or $\max_i |\lambda_i(\widehat{C}W)| < 1$. Thus, there exists a $P\succ 0$ such that for some $\tau\in(0,1)$, we have
\begin{align*}
    (\widehat{C}W)^\top P (\widehat{C}W) \preceq \tau^2 P. 
\end{align*}
We now denote the norm induced by such a $P$ as $|| \cdot ||_P$. Then, $|| \widehat{C}W ||_{P}\le \tau$, meaning $|| (\widehat{C}W)^{t} ||\le \tau^t$. By the norm equivalence between the induced matrix norm and the standard operator norm, we have that
\begin{align*}
    || (\widehat{C}W)^{t} ||_{\mathrm{op}}\le\kappa(P) || (\widehat{C}W)^{t} ||_P\le\kappa(P) \tau^t,
\end{align*}
where $\kappa(P) := \sqrt{\lambda_{\max}(P)/\lambda_{\min}(P)}$ is the condition number of $P$. We now see that the previous sum converges:
\begin{align*}
     D(K) \max_{\widehat{C}\in\mathcal{B}_{\widehat{q}}(f(\theta))} \sum_{t=0}^{\infty} t || (\widehat{C} W)^{t}||_{\mathrm{op}} || (\widehat{C} W)^{(t-1)} ||_{\mathrm{op}}
     &\le  D(K) \kappa(P)^{2} \max_{\widehat{C}\in\mathcal{B}_{\widehat{q}}(f(\theta))} \sum_{t=0}^{\infty} t \tau^{2t-1}\\
     &\le  D(K) \kappa(P)^{2} \frac{\tau}{(1-\tau^2)^2},
\end{align*}
completing the proof.
\end{proof}

The finite LQR case follows immediately as a corollary of the above, stated below for completeness.

\begin{corollary}\label{lemma:coverage_bound_control_disc_finite}[Deterministic, discrete-time, finite-horizon]
    Let $J(K, C) := \sum_{t=0}^{T} (x_t^\top (Q + K^\top R K) x_t)$ with $w=0$. Assume that $\mathcal{P}_{\Theta,C}(C\in\mathcal{B}_{\widehat{q}}(f(\Theta))) \ge 1 - \alpha$. Then:
    \begin{equation*}
        \mathcal{P}_{\Theta,C}\left(0\le \mathcal{R}(\Theta, C) \le 2L \widehat{q} + \Delta_{\mathrm{dom}}(\Theta,C)\right) \ge 1 - \alpha,
    \end{equation*}
    where $L$ is the Lipschitz constant of $J(K,\widehat{C})$ in $\widehat{C}\in\mathcal{B}_{\widehat{q}}(f(\theta))$ under the operator norm. Further, if $\widehat{q} < r(C,K^*(C))$, (see discrete-time in \Cref{def:margin_extended}) $\Delta_{\mathrm{dom}}(\Theta,C)=0$.
\end{corollary}

\section{Deterministic Continuous-Time Regret Analysis}\label{section:coverage_bound_control_cont_det}
The proof follows in much the same manner as the discrete-time case, with modest adjustments to the exact specification of the assumptions regarding the dynamics.

\begin{theorem}\label{lemma:coverage_bound_control_cts_inf}[Deterministic, continuous-time]
    Let $J(K, C) := \int_{0}^{\infty} (x(t)^\top (Q + K^\top R K) x(t)) dt$ for $w=0$. Assume that $\mathcal{P}_{\Theta,C}(C\in\mathcal{B}_{\widehat{q}}(f(\Theta))) \ge 1 - \alpha$. Then, under \Cref{assump:opt_technical_1},
    \begin{equation*}
        \mathcal{P}_{\Theta,C}\left(0\le \mathcal{R}(\Theta, C) \le 2L \widehat{q} + \Delta_{\mathrm{dom}}(\Theta,C)\right) \ge 1 - \alpha,
    \end{equation*}
    where $L$ is the Lipschitz constant of $J(K,\widehat{C})$ in $\widehat{C}\in\mathcal{B}_{\widehat{q}}(f(\theta))$ under the operator norm. Further, if $\widehat{q} < r(C,K^*(C))$, (see continuous-time in \Cref{def:margin_extended}) $\Delta_{\mathrm{dom}}(\Theta,C)=0$.
\end{theorem}
\begin{proof}
We consider any fixed $\theta$ and demonstrate the desired properties that $J(K, C)$ is non-negative and $L$-Lipschitz in $\widehat{C}\in\mathcal{B}_{\widehat{q}}(f(\theta))$ under the operator norm for any $K\in\mathcal{K}(\mathcal{B}_{\widehat{q}}(f(\theta)))$, from which \Cref{lemma:coverage_bound} can be invoked to arrive at the desired conclusion. Given the assumed determinism of the dynamics, we further have that $x(t) = e^{\widehat{C}Wt}x(0)$, meaning the above objective setup can equivalently be expressed with the uncertainty sets related to the objective function, namely as:
\begin{equation}
    \int_{0}^{\infty} x(0)^\top (e^{\widehat{C}Wt})^{\top} (Q + K^\top R K) (e^{\widehat{C}Wt})x(0) dt.
\end{equation}
$J$ is clearly non-negative by construction. It, therefore, suffices to demonstrate this objective is Lipschitz continuous with an appropriate Lipschitz constant. We again proceed by bounding the norm of the gradient as follows:
\begin{align*}
    \nabla_{\widehat{C}} \left(\int_{0}^{\infty} x(0)^\top (e^{\widehat{C}Wt})^{\top} (Q + K^\top R K) (e^{\widehat{C}Wt})x(0) dt \right)
    &= \int_{0}^{\infty} t \mathrm{diag}((Q+K^\top  R K) e^{\widehat{C} W t} x(0)) e^{\widehat{C} W t} \mathrm{diag}(x(0)) W^\top dt \\
    &+ \int_{0}^{\infty} t \mathrm{diag}((Q^\top +(R K)^\top  K) e^{\widehat{C} W t} x(0)) e^{\widehat{C} W t} \mathrm{diag}(x(0)) W^\top dt
\end{align*}
We again bound each of these two terms separately, as follows: 
\begin{align*}
    \max_{\widehat{C}\in\mathcal{B}_{\widehat{q}}(f(\theta))} &\left|\left| \int_{0}^{\infty} t \mathrm{diag}((Q+K^\top  R K) e^{\widehat{C} W t} x(0)) e^{\widehat{C} W t} \mathrm{diag}(x(0)) W^\top dt \right|\right|_{\mathrm{op}} \\
    &\le \max_{\widehat{C}\in\mathcal{B}_{\widehat{q}}(f(\theta))} \int_{0}^{\infty} t \left|\left| \mathrm{diag}((Q+K^\top  R K) e^{\widehat{C} W t} x(0)) \right|\right|_{\mathrm{op}} \left|\left| e^{\widehat{C} W t}\right|\right|_{\mathrm{op}} \left|\left| \mathrm{diag}(x(0))\right|\right|_{\mathrm{op}} \left|\left| W^\top\right|\right|_{\mathrm{op}} dt \\
    &= \max_{\widehat{C}\in\mathcal{B}_{\widehat{q}}(f(\theta))} \int_{0}^{\infty} t \left|\left| (Q+K^\top  R K) e^{\widehat{C} W t} x(0) \right|\right|_{\infty} \left|\left| e^{\widehat{C} W t}\right|\right|_{\mathrm{op}} \left|\left| x(0)\right|\right|_{\infty} \left|\left| W\right|\right|_{\mathrm{op}} dt \\
   &\le \max_{\widehat{C}\in\mathcal{B}_{\widehat{q}}(f(\theta))} \int_{0}^{\infty} t (\sqrt{n} \left|\left| Q+K^\top  R K \right|\right|_{\infty} \left|\left| W\right|\right|_{\mathrm{op}} \left|\left| x(0) \right|\right|^{2}_{\infty}) \left|\left| e^{\widehat{C} W t} \right|\right|^{2}_{\mathrm{op}} dt.
\end{align*}
Collecting all terms independent of $t$ into a constant $D(K) = \sqrt{n} || Q+K^\top  R K ||_{\infty} || W||_{\mathrm{op}} || x(0) ||^{2}_{\infty}$ and using the bound $||e^{\widehat{C} W t}||\le \beta(\widehat{C}) e^{-\alpha(\widehat{C}) t}$, we reach the conclusion as:
\begin{align*}
    = \max_{\widehat{C}\in\mathcal{B}_{\widehat{q}}(f(\theta))} D(K) \int_{0}^{\infty} t \left|\left| e^{\widehat{C} W t} \right|\right|^{2}_{\mathrm{op}} dt
    \le D(K) \beta(\widehat{C})^{2} \int_{0}^{\infty} t e^{-2\alpha(\widehat{C}) t} dt
    = \max_{\widehat{C}\in\mathcal{B}_{\widehat{q}}(f(\theta))} \frac{D(K)\beta(\widehat{C})^{2}}{4 \alpha(\widehat{C})^2}, 
\end{align*}
as desired.
\end{proof}
Once again, the proof in the finite time horizon case follows equivalently.

\begin{corollary}\label{lemma:coverage_bound_control_cts_finite}[Deterministic, continuous-time, finite-horizon]
    Let $J(K, C) := \int_{0}^{T} (x(t)^\top (Q + K^\top R K) x(t)) dt$ for $w= 0$. Assume that $\mathcal{P}_{\Theta,C}(C\in\mathcal{B}_{\widehat{q}}(f(\Theta))) \ge 1 - \alpha$. Then, under \Cref{assump:opt_technical_1},
    \begin{equation*}
        \mathcal{P}_{\Theta,C}\left(0\le \mathcal{R}(\Theta, C) \le 2L \widehat{q} + \Delta_{\mathrm{dom}}(\Theta,C)\right) \ge 1 - \alpha,
    \end{equation*}
    where $L$ is the Lipschitz constant of $J(K,\widehat{C})$ in $\widehat{C}\in\mathcal{B}_{\widehat{q}}(f(\theta))$ under the operator norm. Further, if $\widehat{q} < r(C,K^*(C))$, (see continuous-time in \Cref{def:margin_extended}) $\Delta_{\mathrm{dom}}(\Theta,C)=0$.
\end{corollary}

\section{Stochastic Discrete-Time Regret Analysis}\label{section:coverage_bound_control_disc_stoch}
We introduce below the discrete-time analog of the continuous-time decay assumption made in \Cref{assump:opt_technical_2} below.

\begin{assumption}\label{assump:opt_technical_3}
    For any $\theta$, $\exists$ constants $\alpha_{1},\beta_{1} >0$ such that for all $\widehat{C}\in\mathcal{B}_{\widehat{q}}(f(\theta))$, $K\in\mathcal{K}(\widehat{C})$, and $t\ge 0$, $||Q(t) + K^\top R(t) K|| \le \beta_{1} e^{-\alpha_{1} t}$ and $\min_{\widehat{C}\in\mathcal{B}_{\widehat{q}}(f(\theta))} (2 \alpha_{2}(\widehat{C}) + \alpha_{1}) > 0$ where $\alpha_{2}(\widehat{C}) := \inf_{K\in\mathcal{K}(\widehat{C})}(-\log \rho(\widehat{C}W)) >0$.
\end{assumption}

\begin{theorem}\label{lemma:coverage_bound_control_disc_inf_stoch}[Stochastic, discrete-time]
    Let $J(K, C) := \sum_{t=0}^{\infty} (x_t^\top (Q + K^\top R K) x_t)$ with $w_{t}\sim\mathcal{N}(0, \Sigma)$ i.i.d. across $t$ such that $D_{2}(K) := || \Sigma ||_{\mathrm{op}} || W ||_{\mathrm{op}} < \infty$. Assume further that $\mathcal{P}_{\Theta,C}(C\in\mathcal{B}_{\widehat{q}}(f(\Theta))) \ge 1 - \alpha$. Then, under \Cref{assump:opt_technical_1} and \Cref{assump:opt_technical_3},
    \begin{equation*}
        \mathcal{P}_{\Theta,C}\left(0\le \mathcal{R}(\Theta, C) \le 2L \widehat{q} + \Delta_{\mathrm{dom}}(\Theta,C)\right) \ge 1 - \alpha,
    \end{equation*}
    where $L$ is the Lipschitz constant of $J(K,\widehat{C})$ in $\widehat{C}\in\mathcal{B}_{\widehat{q}}(f(\theta))$ under the operator norm. Further, if $\widehat{q} < r(C,K^*(C))$, (see discrete-time in \Cref{def:margin_extended}) $\Delta_{\mathrm{dom}}(\Theta,C)=0$.
\end{theorem}
\begin{proof}
We again consider any fixed $\theta$ and demonstrate the desired properties that $J(K, \widehat{C})$ is non-negative and $L$-Lipschitz in $\widehat{C}\in\mathcal{B}_{\widehat{q}}(f(\theta))$ under the operator norm for any $K\in\mathcal{K}(\mathcal{B}_{\widehat{q}}(f(\theta)))$, from which \Cref{lemma:coverage_bound} can be invoked to arrive at the desired conclusion. Notice the objective can be reformulated in the standard manner as follows:
\begin{align*}
     J(K, \widehat{C}) 
     &:= \mathbb{E}[\sum_{t=0}^{\infty} (x_t^\top (Q_{t} + K^\top R_{t} K) x_t)]
     = \sum_{t=0}^{\infty} \mathbb{E}[(x_t^\top (Q_{t} + K^\top R_{t} K) x_t)] \\
     &= \sum_{t=0}^{\infty} \mathbb{E}[\mathrm{Tr}(x_t^\top (Q_{t} + K^\top R_{t} K) x_t)]
     = \sum_{t=0}^{\infty} \mathbb{E}[\mathrm{Tr}((Q_{t} + K^\top R_{t} K) x_tx_t^\top)]] \\
     &= \sum_{t=0}^{\infty} \mathrm{Tr}((Q_{t} + K^\top R_{t} K) \mathbb{E}[x_tx_t^\top])
     = \sum_{t=0}^{\infty} \mathrm{Tr}((Q_{t} + K^\top R_{t} K) (\mathbb{E}[x_t]\mathbb{E}[x_t]^\top + \mathrm{Var}(x_{t}))).
\end{align*}
We now use the following computations to evaluate this final expression:
\begin{align*}
     \mathbb{E}[x_t] 
     &= \mathbb{E}[(\widehat{C}W)x_{t-1} + w_{t}] \\
     &= \mathbb{E}[(\widehat{C}W)x_{t-1}] + \mathbb{E}[w_{t}]
     = \mathbb{E}[(\widehat{C}W)((\widehat{C}W)x_{t-2} + w_{t-1})] \\
     &= \mathbb{E}[(\widehat{C}W)^{2}x_{t-2}] + (\widehat{C}W)\mathbb{E}[w_{t-1})]
     = ...
     = (\widehat{C}W)^{t}x_{0}.
\end{align*}
\begin{align*}
     \mathrm{Var}(x_{t})
     &= \mathrm{Var}((\widehat{C}W)x_{t-1} + w_{t})
     = (\widehat{C}W)\mathrm{Var}(x_{t-1})(\widehat{C}W)^{\top} + \Sigma \\
     &= (\widehat{C}W)\mathrm{Var}(x_{t-2})(\widehat{C}W)^{\top} + (\widehat{C}W)\Sigma(\widehat{C}W)^{\top} + \Sigma
     = ...
     = \sum_{k=0}^{t-1} (\widehat{C}W)^{k}\Sigma(\widehat{C}W)^{k\top}.
\end{align*}
With these simplifications, we are left with:
\begin{align*}
     J(K, \widehat{C}) &= \sum_{t=0}^{\infty} x_{0}^{\top}(\widehat{C}W)^{t\top}(Q_{t} + K^\top R_{t} K)(\widehat{C}W)^{t}x_{0}  \\
     &+ \sum_{t=0}^{\infty} \sum_{k=0}^{t-1} \mathrm{Tr}((Q_{t} + K^\top R_{t} K)(\widehat{C}W)^{k}\Sigma(\widehat{C}W)^{k\top}))
\end{align*}
$J$ is clearly non-negative by construction. We demonstrate this quantity is Lipschitz continuous with an appropriate Lipschitz constant, again by bounding the gradient. The bound for the first term was demonstrated in the proof of \Cref{lemma:coverage_bound_control_disc_inf}, for which reason we solely present that of the second term as follows:
\begin{align*}
    &\sum_{t=0}^{\infty} \sum_{k=0}^{t-1} \nabla_{\widehat{C}} \mathrm{Tr}((Q_{t} + K^\top R_{t} K)(\widehat{C}W)^{k}\Sigma(\widehat{C}W)^{k\top})) \\
    &= \sum_{t=0}^{\infty} \sum_{k=0}^{t-1} k ((((Q_{t}+K^\top  R_{t} K)^\top (\widehat{C} W)^{k} \Sigma^\top )\odot (\widehat{C} W)^{(k-1)}) W^\top + \\
    &\sum_{t=0}^{\infty} \sum_{k=0}^{t-1} k (((Q_{t}+K^\top  R_{t} K) (\widehat{C} W)^{k} \Sigma)\odot (\widehat{C} W)^{(k-1)}) W^\top 
\end{align*}
We now bound each of these two terms separately, although the structure of the two is the same, so we explicitly show steps for bounding the first, from which the same can be repeated on the second.
\begin{align*}
    &\max_{\widehat{C}\in\mathcal{B}_{\widehat{q}}(f(\theta))} || \sum_{t=0}^{\infty} \sum_{k=0}^{t-1} k ((((Q_{t}+K^\top  R_{t} K)^\top (\widehat{C} W)^{k} \Sigma^\top )\odot (\widehat{C} W)^{(k-1)}) W^\top ||_{\mathrm{op}} \\
    &\le \max_{\widehat{C}\in\mathcal{B}_{\widehat{q}}(f(\theta))} \sum_{t=0}^{\infty} \sum_{k=0}^{t-1} k || ((Q_{t}+K^\top  R_{t} K)^\top (\widehat{C} W)^{k} \Sigma^\top )\odot (\widehat{C} W)^{(k-1)} ||_{\mathrm{op}}  || W ||_{\mathrm{op}} \\
    &\le \max_{\widehat{C}\in\mathcal{B}_{\widehat{q}}(f(\theta))} \sum_{t=0}^{\infty} \sum_{k=0}^{t-1} k || Q_{t}+K^\top  R_{t} K ||_{\mathrm{op}} || (\widehat{C} W)^{k} ||_{\mathrm{op}} || \Sigma ||_{\mathrm{op}} || (\widehat{C} W)^{(k-1)} ||_{\mathrm{op}}  || W ||_{\mathrm{op}} \\
    &\le \max_{\widehat{C}\in\mathcal{B}_{\widehat{q}}(f(\theta))} \sum_{t=0}^{\infty} \sum_{k=0}^{t-1} k \beta_{1} e^{-\alpha_{1} t} || (\widehat{C} W)^{k} ||_{\mathrm{op}} || \Sigma ||_{\mathrm{op}} || (\widehat{C} W)^{(k-1)} ||_{\mathrm{op}}  || W ||_{\mathrm{op}} \\
\end{align*} 
We again repeat the proof technique leveraged in \Cref{section:coverage_bound_control_disc_det}, where we first collect all terms independent of $t$ into a constant $D_{2}(K) = || \Sigma ||_{\mathrm{op}} || W ||_{\mathrm{op}} $. By precisely the same argument as presented there, namely that $K$ stabilizes $\widehat{C}$, we have that there is a $P\succ 0$ such that for some $\tau\in(0,1)$, we have
\begin{align*}
    (\widehat{C}W)^\top P (\widehat{C}W) \preceq \tau^2 P. 
    \implies || (\widehat{C}W)^{t} ||_{\mathrm{op}}\le\kappa(P) \tau^t,
\end{align*}
where $\kappa(P)$ is the condition number of $P$. By \Cref{assump:opt_technical_3}, we have that $\alpha_{2}(\widehat{C}) \le -\log(\tau)$ or $e^{-\alpha_{2}(\widehat{C})} \ge \tau$. Therefore, we have that the prior sum is bounded by
\begin{align*}
    \le D_{2}(K) \kappa(P)^2 \max_{\widehat{C}\in\mathcal{B}_{\widehat{q}}(f(\theta))} \sum_{t=0}^{\infty} \sum_{k=0}^{t-1} k \tau^{2k-1}
    &\le D_{2}(K) \kappa(P)^2 \max_{\widehat{C}\in\mathcal{B}_{\widehat{q}}(f(\theta))} \sum_{t=0}^{\infty} \beta_{1} e^{-\alpha_{1} t} \sum_{k=0}^{t-1} k e^{-\alpha_{2}(\widehat{C})(2k-1)} \\
    &\le D_{2}(K) \kappa(P)^2 \beta_{1} \max_{\widehat{C}\in\mathcal{B}_{\widehat{q}}(f(\theta))} \sum_{k=0}^{t-1} k e^{-\alpha_{2}(\widehat{C})(2k-1)}  \sum_{t=k}^{\infty} e^{-\alpha_{1} t} \\
    &= D_{2}(K) \kappa(P)^2 \beta_{1} \max_{\widehat{C}\in\mathcal{B}_{\widehat{q}}(f(\theta))} \sum_{k=0}^{t-1} k e^{-\alpha_{2}(\widehat{C})(2k-1)} \left(\frac{e^{\alpha_1-\alpha_1 k}}{e^{\alpha_1}-1} \right) \\
    &= D_{2}(K) \kappa(P)^2 \beta_{1}  \left(\frac{1}{e^{\alpha_1}-1} \right) \max_{\widehat{C}\in\mathcal{B}_{\widehat{q}}(f(\theta))} \sum_{k=0}^{t-1} k (e^{-\alpha_{2}(\widehat{C})})^{(2k-1)} (e^{-\alpha_1})^{k-1} \\
    &= D_{2}(K) \kappa(P)^2 \beta_{1}  \left(\frac{1}{e^{\alpha_1}-1} \right) \max_{\widehat{C}\in\mathcal{B}_{\widehat{q}}(f(\theta))} e^{-\alpha_{2}(\widehat{C})} \sum_{k=0}^{t-1} k (e^{-(2\alpha_{2}(\widehat{C}) + \alpha_1)})^{k-1} \\
    &= D_{2}(K) \kappa(P)^2 \beta_{1}  \left(\frac{1}{e^{\alpha_1}-1} \right) \max_{\widehat{C}\in\mathcal{B}_{\widehat{q}}(f(\theta))} \frac{e^{-\alpha_2(\widehat C)}}{(1 - e^{-(2\alpha_2(\widehat C)+\alpha_1)})^2},
\end{align*} 
thus completing the proof.
\end{proof}


\section{Stochastic Continuous-Time Regret Analysis}\label{section:coverage_bound_control_cont_stoch}
\begin{theorem}\label{lemma:coverage_bound_control_cont_inf_stoch}[Stochastic, continuous-time]
    Let $J(K, C) := \mathbb{E}\left[\int_{0}^{\infty} (x(t)^\top (Q(t) + K^\top R(t) K) x(t)) dt \right]$ with $w(t)$ a white noise process with spectral density $\Sigma$ such that $D_{2}(K) := || \Sigma ||_{\mathrm{op}} || W ||_{\mathrm{op}} < \infty$. Assume further that $\mathcal{P}_{\Theta,C}(C\in\mathcal{B}_{\widehat{q}}(f(\Theta))) \ge 1 - \alpha$. Then, under \Cref{assump:opt_technical_1} and \Cref{assump:opt_technical_2},
    \begin{equation*}
        \mathcal{P}_{\Theta,C}\left(0\le \mathcal{R}(\Theta, C) \le 2L \widehat{q} + \Delta_{\mathrm{dom}}(\Theta,C)\right) \ge 1 - \alpha,
    \end{equation*}
    where $L$ is the Lipschitz constant of $J(K,\widehat{C})$ in $\widehat{C}\in\mathcal{B}_{\widehat{q}}(f(\theta))$ under the operator norm. Further, if $\widehat{q} < r(C,K^*(C))$, (see continuous-time in \Cref{def:margin_extended}) $\Delta_{\mathrm{dom}}(\Theta,C)=0$.
\end{theorem}
\begin{proof}
We again consider any fixed $\theta$ and demonstrate the desired properties that $J(K, \widehat{C})$ is non-negative and $L$-Lipschitz in $\widehat{C}\in\mathcal{B}_{\widehat{q}}(f(\theta))$ under the operator norm for any $K\in\mathcal{K}(\mathcal{B}_{\widehat{q}}(f(\theta)))$, from which \Cref{lemma:coverage_bound} can be invoked to arrive at the desired conclusion. Notice the objective can be reformulated in the standard manner as follows:
\begin{align*}
     J(K, \widehat{C}) 
     &:= \mathbb{E}\left[\int_{0}^{\infty} (x(t)^\top (Q(t) + K^\top R(t) K) x(t)) dt \right] \\
     &= \int_{0}^{\infty} \mathbb{E}[(x(t)^\top (Q(t) + K^\top R(t) K) x(t))] dt
     = \int_{0}^{\infty} \mathbb{E}[\mathrm{Tr}(x(t)^\top (Q(t) + K^\top R(t) K) x(t))] dt \\
     &= \int_{0}^{\infty} \mathbb{E}[\mathrm{Tr}((Q(t) + K^\top R(t) K) x(t)x(t)^\top)]] dt
     = \int_{0}^{\infty} \mathrm{Tr}((Q(t) + K^\top R(t) K) \mathbb{E}[x(t)x(t)^\top]) dt \\
     &= \int_{0}^{\infty} \mathrm{Tr}((Q(t) + K^\top R(t) K) (\mathbb{E}[x(t)]\mathbb{E}[x(t)]^\top + \mathrm{Var}(x(t)))) dt.
\end{align*}
We now obtain the expressions for $\mathbb{E}[x(t)]$ and $\mathrm{Var}(x(t))$ using standard results from stochastic differential equations. For a full review on this topic, see \citep{sarkka2019applied}:
\begin{align*}
     J(K, \widehat{C}) &= \int_{0}^{\infty} x(0)^\top (e^{\widehat{C}Wt})^{\top} (Q(t) + K^\top R(t) K) (e^{\widehat{C}Wt})x(0) dt \\
     &+ \int_{0}^{\infty} \int_{0}^{t} \mathrm{Tr}((Q(t) + K^\top R(t) K) e^{\widehat{C}Wk} \Sigma e^{\widehat{C}Wk\top}) dk dt
\end{align*}
$J$ is clearly non-negative by construction. We demonstrate this quantity is Lipschitz continuous with an appropriate Lipschitz constant, again by bounding the gradient in much the same manner as the above bound. The bound for the first term was demonstrated in the proof of \Cref{lemma:coverage_bound_control_disc_inf}, which holds under the finiteness assumption of $D_{1}(K)$. We, thus, solely present that of the second term as follows:
\begin{align*}
    &\int_{0}^{\infty} \int_{0}^{t} \nabla_{\widehat{C}} \mathrm{Tr}((Q(t) + K^\top R(t) K) e^{(\widehat{C}W)k} \Sigma e^{(\widehat{C}W)k\top}) dk dt \\
    &= \int_{0}^{\infty} \int_{0}^{t} k (((Q(t)^\top +(R(t) K)^\top  K) e^{k \widehat{C} W} \Sigma^\top )\odot e^{k \widehat{C} W}) W^\top dk dt \\  
    &+\int_{0}^{\infty} \int_{0}^{t} k (((Q(t)+K^\top  R(t) K) e^{k \widehat{C} W} \Sigma)\odot e^{k \widehat{C} W}) W^\top
\end{align*}
We now bound each of these two terms separately, although the structure of the two is the same, so we explicitly show steps for bounding the first, from which the same can be repeated on the second.
\begin{align*}
    L 
    &\le \max_{\widehat{C}} || \int_{0}^{\infty} \int_{0}^{t} k (((Q(t)^\top +(R(t) K)^\top  K) e^{k \widehat{C} W} \Sigma^\top )\odot e^{k \widehat{C} W}) W^\top dk dt ||_{\mathrm{op}} \\
    &\le \max_{\widehat{C}} \int_{0}^{\infty} \int_{0}^{t} k || Q(t)+K^\top  R(t) K ||_{\mathrm{op}}  || \Sigma ||_{\mathrm{op}}  || W ||_{\mathrm{op}} || e^{k \widehat{C} W}||^{2}_{\mathrm{op}} dk dt
\end{align*} 
We again now collect all terms independent of $t$ into a constant $D_{2}(K) = || \Sigma ||_{\mathrm{op}} || W ||_{\mathrm{op}} $, leaving
\begin{align*}
    &\le \max_{\widehat{C}} D_{2}(K) \int_{0}^{\infty} \beta_{1} e^{-\alpha_{1} t} \int_{0}^{t} k \beta_{2}(\widehat{C})^{2} e^{-2\alpha_{2}(\widehat{C}) k} dk dt \\
    &= \max_{\widehat{C}} \frac{D_{2}(K) \beta_{1} \beta_{2}(\widehat{C})^{2}}{4\alpha_{2}^{2}(\widehat{C})} \int_{0}^{\infty} e^{-\alpha_{1} t} \left(1 -2\alpha_{2}(\widehat{C}) t e^{-2\alpha_{2}(\widehat{C}) t} - e^{-2\alpha_{2}(\widehat{C}) t}\right) dt \\
    &= \max_{\widehat{C}} \frac{D_{2}(K) \beta_{1} \beta_{2}(\widehat{C})^{2}}{4\alpha_{2}^{2}(\widehat{C})}
    \left(\frac{1}{\alpha_{1}} - \frac{2\alpha_{2}(\widehat{C})}{(\alpha_{1} + 2\alpha_{2}(\widehat{C}))^2} - \frac{1}{\alpha_{1} + 2\alpha_{2}(\widehat{C})}\right)
    = \max_{\widehat{C}} \frac{D_{2}(K) \beta_{1} \beta_{2}(\widehat{C})^{2}}{\alpha_{1}(\alpha_{1} + 2\alpha_{2}(\widehat{C}))^2}
\end{align*} 
We, therefore, again have the desired upper bound on the Lipschitz constant, as desired.
\end{proof}

\section{Unimodal Assumption Explanation}\label{section:unimodal_assump}
In classical engineering design, one would prescribe the dynamics of the system by explicitly writing out the physics of the system; this is so universally done that it may not even feel like an assumption in engineering design. This is the sense in which we mean that the design parameters commonly have some ``unimodal'' (often Dirac) measure in mapping to the system dynamics. For instance, if one is studying a cart pole system with a position $x$ and angle $\theta$, a common characterization (see Chapter 3 of \citep{underactuated}) is given by:
\begin{align*}
    \ddot{x} =& \frac{1}{m_c + m_p \sin^2\theta}\left[ f_x + m_p \sin\theta (l \dot\theta^2 + g\cos\theta)\right] \\
    \ddot{\theta} =& \frac{1}{l(m_c + m_p \sin^2\theta)} \left[ -f_x
    \cos\theta - m_p l \dot\theta^2 \cos\theta \sin\theta - (m_c + m_p) g \sin\theta \right] 
\end{align*}
Here, the parameters $m_c$, $m_p,$ and $l$ could all be viewed as the ``design parameters'' $\theta$ depending on what one, as an engineer, has control over. Knowing that the underlying reality can be described by single set of dynamics is, therefore, what motivates using a unimodal model, as we wished to highlight with the models used in previous UCCD works. 

\section{Coverage Guarantees Under Noisy Observations}\label{section:coveragefornoisyobs}

\begin{theorem}
    Let $\widetilde{C} = C + \epsilon$ where $\text{vec}(\epsilon) \sim \mathcal{N}(0, \Sigma)$, where $\epsilon\indep (\Theta, C)$. Assume $\mathcal{U}(\theta) = \{C' \mid || f(\theta) - C' ||_{\mathrm{op}} \le \widehat{q}\}$ satisfies $\mathcal{P}_{\Theta,\widetilde{C}}(\widetilde{C} \in \mathcal{U}(\Theta)) \ge 1 - \alpha$, where $|| \cdot ||_{\mathrm{op}}$ denotes the matrix operator norm. If for any $\theta\in\Theta$ and $\delta > 0$, $\mathcal{P}(\widehat{q}^2 - \delta \leq \|C - f(\theta)\|_{\mathrm{op}}^{2}   \leq \widehat{q}^2 \mid\Theta = \theta) > \mathcal{P}(\widehat{q}^2 \leq \|C - f(\theta)\|_{\mathrm{op}}^{2}   \leq \widehat{q}^2 + \delta \mid\Theta = \theta)$, then
    \begin{align*}
        \mathcal{P}_{\Theta, C}(C \in \mathcal{U}(\Theta)) \geq \mathcal{P}_{\Theta,\widetilde{C}}(\widetilde{C} \in \mathcal{U}(\Theta)) \geq 1 - \alpha.
    \end{align*}
\end{theorem}

\begin{proof}
    Given that $\mathcal{P}_{\Theta,\widetilde{C}}(\widetilde{C} \in \mathcal{U}(\Theta)) \geq 1 - \alpha$, it suffices to show that $\mathcal{P}(C \in \mathcal{U}(\theta) \mid\Theta = \theta) \geq \mathcal{P}(\widetilde{C} \in \mathcal{U}(\theta)  \mid\Theta = \theta)$ for all $\theta$, as the conclusion can be drawn by the law of total probability:
    \begin{align*}
        & \mathcal{P}(\widetilde{C} \in \mathcal{U}(\Theta)\mid\Theta = \theta) \\
        &= \mathcal{P}\left(\left\|\widetilde{C} - f(\theta)\right\|_{\mathrm{op}}^{2} \leq \widehat{q}^2\mid\Theta = \theta\right) \\
        &= \mathcal{P}\left(\underset{\|x\| = 1}{\sup}\left\{\left\|Cx + \epsilon x - f(\theta)x\right\|_{2}^{2} \right\}\leq \widehat{q}^2\mid\Theta = \theta\right) \\
        &= \mathcal{P}\left(\underset{\|x\| = 1}{\sup}\left\{\left\|Cx - f(\theta)x\right\|_{2}^{2} + 2x^T\epsilon^T\left(Cx - f(\theta)x\right)   + \|\epsilon x\|_{2}^{2}\right\} \leq \widehat{q}^2\mid\Theta = \theta\right)
    \end{align*}
    We now \textit{lower} bound this inner quantity, from which 
    \begin{align*}
        \underset{\|x\| = 1}{\sup}& \left\{\left\|Cx - f(\theta)x\right\|_{2}^{2} + 2x^T\epsilon^T\left(Cx - f(\theta)x\right)   + \|\epsilon x\|_{2}^{2}\right\} \\
        &\ge \underset{\|x\| = 1}{\sup}\left\{\left\|Cx - f(\theta)x\right\|_{2}^{2} + 2x^T\epsilon^T\left(Cx - f(\theta)x\right) \right\} \\
        &\ge \left\|Cx' - f(\theta)x'\right\|_{2}^{2} + 2(x')^T\epsilon^T\left(Cx' - f(\theta)(x')\right),
    \end{align*}
    for any choice of $x' : || x' ||_2=1$. The second line follows by the trivial fact that $\|\epsilon x\|_{2}^{2}\ge0$ and the third from the fact that the previous line is the supremum of \textit{all} such possible values $x'$. We now specifically take $x' := \argmax_{x} \left\|Cx - f(\theta)x\right\|_{2}^{2}$ and denote it as $x^*$. From here, we arrive at the final bound
    \begin{align*}
        \underset{\|x\| = 1}{\sup}&\left\{\left\|Cx - f(\theta)x\right\|_{2}^{2} + 2x^T\epsilon^T\left(Cx - f(\theta)x\right)   + \|\epsilon x\|_{2}^{2}\right\} \\
        &\ge \left\|(C - f(\theta))(x^*)\right\|_{2}^{2} + 2(x^*)^T\epsilon^T\left(Cx^* - f(\theta)(x^*)\right)\\
        &=: \left\|C - f(\theta)\right\|_{\mathrm{op}}^{2} + 2(x^*)^T\epsilon^T\left(Cx^* - f(\theta)(x^*)\right)
    \end{align*}
    Since this is a \textit{lower} bound on the original quantity of interest, we have that the probability this quantity is upper bounded by $\widehat{q}^2$ is \textit{greater} than that of the original quantity being upper bounded. That is, 
    \begin{align*}
        \mathcal{P}&\left(\underset{\|x\| = 1}{\sup}\left\{\left\|Cx - f(\theta)x\right\|_{2}^{2} + 2x^T\epsilon^T\left(Cx - f(\theta)x\right)   + \|\epsilon x\|_{2}^{2}\right\} \leq \widehat{q}^2\mid\Theta = \theta\right) \\
        &\leq \mathcal{P}\left(\left\|C - f(\theta)\right\|_{\mathrm{op}}^{2} + 2x^{*^T}\epsilon^T\left(Cx^* - f(\theta)x^*\right)  \leq \widehat{q}^2\mid\Theta = \theta\right)
    \end{align*}
    Let $\delta = 2x^{*^T}\epsilon^T\left(Cx^* - f(\theta)x^*\right)$. Then: 
    \begin{align*}
        &:= \mathcal{P}\left(\left\|C - f(\theta)\right\|_{\mathrm{op}}^{2} + \delta  \leq \widehat{q}^2\mid\Theta = \theta\right) 
        \\
        &= \mathcal{P}\left(\left\|C - f(\theta)\right\|_{\mathrm{op}}^{2} + \delta  \leq \widehat{q}^2 \mid\Theta = \theta, \delta > 0\right)\mathcal{P}(\delta > 0\mid\Theta = \theta) \\
        & \quad+ \mathcal{P}\left(\left\|C - f(\theta)\right\|_{\mathrm{op}}^{2} + \delta  \leq \widehat{q}^2 \mid\Theta = \theta, \delta \leq 0\right)\mathcal{P}(\delta \leq 0\mid\Theta = \theta) \\
        &= \mathcal{P}\left(\left\|C - f(\theta)\right\|_{\mathrm{op}}^{2}   \leq \widehat{q}^2 - \delta \mid\Theta = \theta, \delta > 0\right)\mathcal{P}(\delta > 0\mid\Theta = \theta) \\
        & \quad+ \mathcal{P}\left(\left\|C - f(\theta)\right\|_{\mathrm{op}}^{2}  \leq \widehat{q}^2 + \delta \mid\Theta = \theta, \delta > 0\right)\mathcal{P}(\delta \leq 0\mid\Theta = \theta)
    \end{align*}
    In this final line, we made use of the fact that
    \begin{align*}
        \mathcal{P}\left(\left\|C - f(\theta)\right\|_{\mathrm{op}}^{2}  \leq \widehat{q}^2 - \delta \mid\Theta = \theta, \delta \le 0\right) 
        &= \mathcal{P}\left(\left\|C - f(\theta)\right\|_{\mathrm{op}}^{2}  \leq \widehat{q}^2 + \delta \mid\Theta = \theta, \delta > 0\right),
    \end{align*}
    which follows since the distribution of $\delta$ is symmetric about 0 by the symmetry of the distribution of $\epsilon$. In particular, $\delta = f(C,\epsilon)$; since $C\indep\epsilon$, the joint distributions $\mathcal{P}(C,\epsilon)$ and $\mathcal{P}(C,-\epsilon)$ are identical. Thus, the distribution of $\delta' = f(C,-\epsilon)$ matches that of $\delta$. Using this, we add terms that sum to 0 as follows:
    \begin{align*}
       &\left.
        \begin{aligned}
        & =\mathcal{P}\!\left(\left\|C - f(\theta)\right\|_{\mathrm{op}}^{2} \leq \widehat{q}^2 \mid\Theta = \theta\right) \\
        &\quad -\mathcal{P}\!\left(\delta >0\mid\Theta = \theta\right)\mathcal{P}\!\left(\left\|C - f(\theta)\right\|_{\mathrm{op}}^{2} \leq \widehat{q}^2 \mid\Theta = \theta,\delta>0 \right) \\
        &\quad -\mathcal{P}\!\left(\delta \leq 0\mid\Theta = \theta\right)\mathcal{P}\!\left(\left\|C - f(\theta)\right\|_{\mathrm{op}}^{2} \leq \widehat{q}^2 \mid\Theta = \theta,\delta \leq 0\right)
        \end{aligned}
        \right\}
        = 0 \\
        &\quad+ \mathcal{P}\left(\left\|C - f(\theta)\right\|_{\mathrm{op}}^{2}   \leq \widehat{q}^2 - \delta \mid\Theta = \theta, \delta > 0\right)\mathcal{P}(\delta > 0\mid\Theta = \theta) \\
        & \quad+ \mathcal{P}\left(\left\|C - f(\theta)\right\|_{\mathrm{op}}^{2}  \leq \widehat{q}^2 + \delta \mid\Theta = \theta, \delta > 0\right)\mathcal{P}(\delta \leq 0\mid\Theta = \theta),
    \end{align*}
    where these newly added terms will be used for manipulation subsequently. From here, we re-express this expression with expectations, where we again use the symmetry in $\delta$ in the second term:
    \begin{align*}
        &= \mathcal{P}\left(\left\|C - f(\theta)\right\|_{\mathrm{op}}^{2}   \leq \widehat{q}^2 \mid\Theta = \theta\right)\\
        & \quad -\mathcal{P}\left(\delta >0\mid\Theta = \theta\right)\mathbb{E}\left[\mathcal{P}(\left\|C - f(\theta)\right\|_{\mathrm{op}}^{2}   \leq \widehat{q}^2 \mid\Theta = \theta)\mid\delta>0 \right] \\
        & \quad - \mathcal{P}\left(\delta \leq 0\mid\Theta = \theta\right)\mathbb{E}\left[\mathcal{P}(\left\|C - f(\theta)\right\|_{\mathrm{op}}^{2}   \leq \widehat{q}^2 \mid\Theta = \theta)\mid\delta > 0\right] \\
        & \quad +\mathcal{P}\left(\delta >0\mid\Theta = \theta\right)\mathbb{E}\left[\mathcal{P}\left(\left\|C - f(\theta)\right\|_{\mathrm{op}}^{2}   \leq \widehat{q}^2 - \delta\mid \Theta = \theta \right) \mid \delta > 0\right] \\
        & \quad +\mathcal{P}\left(\delta \leq 0\mid \Theta = \theta\right)\mathbb{E}\left[\mathcal{P}\left(\left\|C - f(\theta)\right\|_{\mathrm{op}}^{2}   \leq \widehat{q}^2 + \delta \mid\Theta = \theta\right)\mid\delta > 0\right].
    \end{align*}
    With this rewrite, we can group terms and conclude using the stated assumption on the ``peaking'' structure of the probability in the prediction region:
    \begin{align*}
        &= \mathcal{P}\left(\left\|C - f(\theta)\right\|_{\mathrm{op}}^{2}   \leq \widehat{q}^2 \mid \Theta = \theta\right) \\
        & \quad- \mathcal{P}\left(\delta >0\mid\Theta = \theta\right)\mathbb{E}\left[\mathcal{P}\left(\left\|C - f(\theta)\right\|_{\mathrm{op}}^{2}   \leq \widehat{q}^2 \mid\Theta = \theta\right) - \mathcal{P}\left(\left\|C - f(\theta)\right\|_{\mathrm{op}}^{2}   \leq \widehat{q}^2 - \delta \mid\Theta = \theta\right)\mid\delta > 0\right] \\
        & \quad+ \mathcal{P}\left(\delta \leq 0\mid\Theta = \theta\right)\mathbb{E}\left[\mathcal{P}\left(\left\|C - f(\theta)\right\|_{\mathrm{op}}^{2}   \leq \widehat{q}^2 + \delta\mid\Theta = \theta \right) - \mathcal{P}\left(\left\|C - f(\theta)\right\|_{\mathrm{op}}^{2}   \leq \widehat{q}^2 \mid\Theta = \theta\right) \mid\delta > 0\right]  \\
        &= \mathcal{P}\left(\left\|C - f(\theta)\right\|_{\mathrm{op}}^{2}   \leq \widehat{q}^2 \mid\Theta = \theta\right) \\
        & \quad- \mathcal{P}\left(\delta >0\mid\Theta = \theta\right)\mathbb{E}\left[\mathcal{P}\left(\widehat{q}^2 - \delta \leq \left\|C - f(\theta)\right\|_{\mathrm{op}}^{2}   \leq \widehat{q}^2 \mid\Theta = \theta\right) \mid\delta > 0\right] \\
        & \quad+ \mathcal{P}\left(\delta \leq 0\mid\Theta = \theta\right)\mathbb{E}\left[\mathcal{P}\left(\widehat{q}^2 \leq \left\|C - f(\theta)\right\|_{\mathrm{op}}^{2}   \leq \widehat{q}^2 + \delta \mid\Theta = \theta\right) \mid \delta > 0\right].
    \end{align*}
    We, therefore, have that $\mathcal{P}(\widetilde{C} \in \mathcal{U}(\theta)\mid\Theta = \theta) \le \mathcal{P}\left(\left\|C - f(\theta)\right\|_{\mathrm{op}}^{2}   \leq \widehat{q}^2 \mid\Theta = \theta\right) + \Delta$, where 
    \begin{align*}
        \Delta &:= \mathcal{P}\left(\delta \leq 0\mid\Theta = \theta\right)\mathbb{E}\left[\mathcal{P}\left(\widehat{q}^2 \leq \left\|C - f(\theta)\right\|_{\mathrm{op}}^{2}   \leq \widehat{q}^2 + \delta \mid\Theta = \theta\right) \mid\delta > 0\right] \\
        & - \mathcal{P}\left(\delta >0\mid\Theta = \theta\right)\mathbb{E}\left[\mathcal{P}\left(\widehat{q}^2 - \delta \leq \left\|C - f(\theta)\right\|_{\mathrm{op}}^{2}   \leq \widehat{q}^2 \mid\Theta = \theta\right)\mid\delta > 0\right]
    \end{align*}
    By the assumption, we know that for all $\delta > 0$:
    \begin{gather*}
        \mathcal{P}\left(\widehat{q}^2 - \delta \leq \left\|C - f(\theta)\right\|_{\mathrm{op}}^{2}   \leq \widehat{q}^2\mid\Theta = \theta \right) > \mathcal{P}\left(\widehat{q}^2 \leq \left\|C - f(\theta)\right\|_{\mathrm{op}}^{2}   \leq \widehat{q}^2 + \delta\mid\Theta = \theta \right).
    \end{gather*}
    We also know that $\mathcal{P}(\delta \leq 0\mid\Theta = \theta) \leq \mathcal{P}(\delta \geq 0\mid\Theta = \theta)$ since
    \begin{align*}
        \mathcal{P}(\delta \leq 0\mid\Theta = \theta) &= \mathcal{P}\left(2x^{*^T}\epsilon^T\left(Cx^* - f(\theta)x^*\right) \leq 0 \mid\Theta = \theta\right) \\
        &= \mathbb{E}_{C \mid \Theta = \theta}\left[\mathcal{P}_{\epsilon\mid C = c, \Theta = \theta}\left(x^{*^T}\epsilon^T\left(c - f(\theta)\right)x^* \leq 0 \right)\right] \\
        &= \mathbb{E}_{C \mid \Theta = \theta}\left[0.5\right] \\
        &= 0.5
    \end{align*}
    Therefore, $\mathcal{P}\left(\left\|\widetilde{C} - f(\theta)\right\|_{\mathrm{op}} \leq \widehat{q}\mid\Theta = \theta\right) - \mathcal{P}\left(\left\| C - f(\theta)\right\|_{\mathrm{op}} \leq \widehat{q}\mid\Theta = \theta\right) \leq \Delta \leq 0$
\end{proof}

\section{LQR $C$ Gradient}\label{section:lyap_c_grad}
We follow the presentation of \citep{fazel2018global} to provide the derivation of $\nabla_{C} J(K, C)$. Note that the following derivation is given for the discrete-time setting; the continuous-time derivation follows in a similar fashion with a modification in the Lyapunov equations.

\begin{lemma}
    Let $J(K, C)$ be the infinite horizon, discrete-time, deterministic analog of that defined in \Cref{eqn:lqr}, i.e. $J(K, C) := \sum_{t=0}^{\infty} (x_{t}^\top (Q + K^\top R K) x_{t})$ for $w=0$. Then,
    \begin{equation}
        \nabla_{C} J(K, C) = 2 P_{K} C W X_{K} W^{T},
    \end{equation}
    where $X_{K}$ and $P_{K}$ respectively solve the following two Lyapunov equations: $\Delta_{K} X_K \Delta^{\top}_{K} - X_K = -X_0$ and $P_{K} = \Delta^{\top}_{K}  P_{K} \Delta_{K} + Q + K^{\top} R K$, where $\Delta_{K} := A - BK$.
\end{lemma}

\begin{proof}
    By the standard reformulation of $J(K, C)$ as described in \citep{fazel2018global}, we can rewrite $J(K, C, x_0) = x_{0}^{\top} P_{K} x_{0}$, where we now make the notational change to make explicit the dependence on $x_0$, as it pertains to the derivation below. We then have that
    \begin{align*}
        J(K, C, x_0)
        &= x_{0}^{\top} \Delta^{\top}_{K}  P_{K} \Delta_{K} x_{0} + x_{0}^{\top} (Q + K^{\top} R K) x_{0} \\
        &= J(K, C, \Delta_{K} x_0) + x_{0}^{\top} (Q + K^{\top} R K) x_{0}.
    \end{align*}
    From here, we have that
    \begin{align*}
        \nabla_{C} J(K, C, x_0)
        &= \nabla_{C} J(K, C, \Delta_{K} x_0) + \nabla_{C} \cancelto{0}{(x_{0}^{\top} (Q + K^{\top} R K) x_{0})} \\
        &= 2 P_{K} C W x_0 x_0^{T} W^{T} + \nabla_{C} J(K, C, \Delta_{K} x_1)\vert_{x_{1} := (A - BK) x_0} \\
        &= ... \\
        &= 2 P_{K} C W (\sum_{t=0}^{\infty} x_t x_t^{T}) W^{T} \\
        &= 2 P_{K} C W X_{K} W^{T},
    \end{align*}
    where the final equality follows from the well-known correspondence between this infinite sum and the aforementioned Lyapunov reformulation.
\end{proof}

\section{Policy Gradient Convergence Guarantee}\label{section:pg_conv_guar}

\begin{lemma}\label{lemma:danskin_grad_dom}
    Suppose $f(x, y)$ is $c(y)$-gradient dominated for any $y\in\mathcal{Y}$, i.e. for any fixed $y$, there is a $c(y)$ such that, for any $x\in\mathcal{X}$ and $x^{*}(y) := \argmin_{x\in\mathcal{X}} f(x, y)$: 
    $$
    f(x,y) - f(x^{*}(y),y) \le c(y) || \nabla_{x} f(x,y) ||_{F}^{2}.
    $$
    Further, let $\phi(x) := \max_{y\in\mathcal{Y}} f(x, y)$ and $x^* := \argmin_{x\in\mathcal{X}} \phi(x)$. Assume that $\Argmax_{y\in\mathcal{Y}} f(x, y)\neq\emptyset$ $\forall x$. Then,
    $$
    \phi(x) - \phi(x^{*})\le c^{*}(x) \min_{y^{*}\in\Argmax_{y\in\mathcal{Y}} f(x, y)} || \nabla_{x} f(x,y^{*}) ||_{F}^{2},
    $$
    where $c^{*}(x) := \sup_{y^{*}\in\Argmax_{y\in\mathcal{Y}} f(x,y)} c(y^{*})$.
\end{lemma}
\begin{proof}
    Note that, for any maximizer $y^{*} \in \Argmax_{y\in\mathcal{Y}} f(x, y)$, we have
    \begin{align*}
        \phi(x) - \phi(x^{*})
        &:= \max_{y\in\mathcal{Y}} f(x, y) - \max_{y\in\mathcal{Y}} f(x^{*}, y)
        = f(x, y^{*}) - \max_{y\in\mathcal{Y}} f(x^{*}, y) \\
        &\le f(x, y^{*}) - f(x^{*}, y^{*})
        \le f(x, y^{*}) - f(x^{*}(y^*), y^{*})
        \le c(y^{*}) || \nabla_{x} f(x,y^{*}) ||_{F}^{2}
    \end{align*}
    Given that this relationship holds for all such $y^{*}$, we immediately get that
    \begin{align*}
        \phi(x) - \phi(x^{*})
        \le \min_{y^{*}\in\Argmax_{y\in\mathcal{Y}} f(x, y)} c(y^{*}) || \nabla_{x} f(x,y^{*}) ||_{F}^{2}
        \le c^{*}(x) \min_{y^{*}\in\Argmax_{y\in\mathcal{Y}} f(x, y)} || \nabla_{x} f(x,y^{*}) ||_{F}^{2},
    \end{align*}
    where $c^{*}(x) := \sup_{y^{*}\in\Argmax_{y\in\mathcal{Y}} f(x, y)} c(y^{*})$. \qedhere
\end{proof}

We now make use of the known fact that $J(K,C)$ is gradient-dominated for any fixed $C$, in turn satisfying the conditions of \Cref{lemma:danskin_grad_dom}, from which we reach the desired conclusion. The former fact was demonstrated in \citep{bu2019lqr}, which we present below for sake of convenience with modification of notational conventions to match that used herein.

\begin{lemma}\label{lemma:grad_dom_lqr}
    (Lemma 5 of \citep{fazel2018global}) Let $J(K, C)$ be the infinite horizon, discrete-time, deterministic analog of that defined in \Cref{eqn:lqr}, i.e. $J(K, C) := \sum_{t=0}^{\infty} (x_{t}^\top (Q + K^\top R K) x_{t})$ for $w=0$. Then, if $X_{K}\succcurlyeq0$ and $K\in\mathcal{K}(C)$,
    \begin{equation}
        J(K,C) - J(K^{*}(C),C)\le\frac{|| X_{K^{*}(C)} ||}{\sigma_{\min}(X_0)^{2} \sigma_{\min}(R)} || \nabla_{K} J(K, C) ||^{2}_{F},
    \end{equation}
    where $X_0\succ 0$ and $X_{K}$ and $P_{K}$ respectively solve the following two equations: $\Delta_{K} X_K \Delta^{\top}_{K} - X_K = -X_0$ and $P_{K} = \Delta^{\top}_{K}  P_{K} \Delta_{K} + Q + K^{\top} R K$, where $\Delta_{K} := A - BK$.
\end{lemma}

\begin{lemma}\label{lemma:grad_dom}
    Let $\phi(K) := \max_{\widehat{C} \in \mathcal{C}} J(K, \widehat{C})$ and $K^{*}_{\mathrm{rob}}(\mathcal{C}) := \argmin_{K\in\mathcal{K}(\mathcal{C})} \phi(K)$, with $J$ 
    the infinite horizon, discrete-time, deterministic analog of that defined in \Cref{eqn:lqr}, i.e. $J(K, C) := \sum_{t=0}^{\infty} (x_{t}^\top (Q + K^\top R K) x_{t})$ for $w=0$.
    Then, for $K\in\mathcal{K}(\mathcal{C})$ where $X_{K}\succcurlyeq0$ for all $\widehat{C} \in \mathcal{C}$, $\phi(K)$ satisfies
    \begin{equation*}
        \phi(K) - \phi(K^{*}_{\mathrm{rob}}(\mathcal{C})) \le \mu^{*}(K) \min_{C^{*}\in\Argmax_{\widehat{C}\in\mathcal{C}} J(K, \widehat{C})} || \nabla_{K} J(K, C^*) ||_{F}^{2}
    \end{equation*}    
    for $\mu^{*}(K) := \sup_{C^{*}\in\Argmax_{\widehat{C}\in\mathcal{C}} J(K, \widehat{C})} \frac{|| X_{K^{*}(C^{*})} ||}{\sigma_{\min}(X_0)^{2} \sigma_{\min}(R)}$,
    where $X_0\succ 0$ and $X_{K}$ and $P_{K}$ respectively solve the following two equations: $\Delta_{K} X_K \Delta^{\top}_{K} - X_K = -X_0$ and $P_{K} = \Delta^{\top}_{K}  P_{K} \Delta_{K} + Q + K^{\top} R K$, where $\Delta_{K} := A - BK$.
\end{lemma}
\begin{proof}
    The proof for this follows immediately by demonstrating the assumption of \Cref{lemma:danskin_grad_dom} is satisfied by \Cref{lemma:grad_dom_lqr}.
\end{proof}

\begin{theorem}\label{thm:opt_conv}
    Let $\phi(K) := \max_{C \in \mathcal{C}} J(K, C)$ and $K^{*}_{\mathrm{rob}}(\mathcal{C}) := \argmin_{K\in\mathcal{K}(\mathcal{C})} \phi(K)$, with $J$ 
    the infinite horizon, discrete-time, deterministic analog of that defined in \Cref{eqn:lqr}, i.e. $J(K, C) := \sum_{t=0}^{\infty} (x_{t}^\top (Q + K^\top R K) x_{t})$ for $w=0$. Let $K^{(t)}$ be the $t$-th iterate of \Cref{alg:crc}. Assume for each iterate $t$, the optimization over $C$ converges, i.e. $C^{(T_C)} = C^{*}(K^{(t)})$, that $K^{(t)}\in\mathcal{K}(\mathcal{C})$, and that $X_{K}\succcurlyeq0$ for all $\widehat{C}\in\mathcal{C}$ and $K\in\mathcal{K}(\mathcal{C})$. Denote $\nu := \min_{\widehat{C}\in\mathcal{C}} \min_{K\in\mathcal{K}(\mathcal{C})} \sigma_{\min}(X_{K})$. If in \Cref{alg:crc} 
    \begin{equation}\label{eqn:learning_rate_assump}
    \begin{aligned}
        \eta_{K}\le\min_{[\widehat{A},\widehat{B}] := \widehat{C} \in \mathcal{C}} \frac{1}{16}\min\Big\{
            \left(\frac{\sigma_{\min}(Q)\nu}{J(K, \widehat{C})}\right)^{2} 
            \frac{1}{|| \widehat{B} || || \nabla_{K} J(K, \widehat{C}) || (1 + || \widehat{A} - \widehat{B}K ||)},\\
            \frac{\sigma_{\min}(Q)}{2 J(K, \widehat{C}) || R + \widehat{B}^{\top}P_{K}\widehat{B} ||}
        \Big\}.
    \end{aligned}
    \end{equation}
    then, there exists a $\gamma > 0$ such that $\phi(K^{(T)}) - \phi(K^{*}_{\mathrm{rob}}(\mathcal{C}))
    \le (1 - \gamma)^{T} (\phi(K_{0}) - \phi(K^{*}_{\mathrm{rob}}(\mathcal{C}))).$
\end{theorem}
\begin{proof}
    We follow the proof strategy developed in \citep{fazel2018global}, specifically in their presentation of Lemma 24, in which we leverage the above developed gradient dominance result, namely that in \Cref{lemma:grad_dom}. We first note that \Cref{alg:crc} is equivalent to performing subgradient descent over $\phi(K)$ if we assume convergence of the inner maximization over $C$, that is if $C^{(T_C)} = C^{*}(K^{(t)})$ for some $C^{*}\in\Argmax_{\widehat{C}\in\mathcal{C}} J(K^{(t)}, \widehat{C})$. It, therefore, suffices to characterize subgradient descent, where $K^{(t+1)} := K^{(t)} - \eta_{K} g_t$. 
    
    To complete this proof, it suffices to demonstrate $\phi(K^{(t)}) - \phi(K^{(t+1)}) \ge \gamma(K^{(t)}) || g_t ||^{2}_{F}$ for some $\gamma(K^{(t)}) > 0$, since this along with subgradient dominance can be used to establish the desired convergence guarantees by first demonstrating this intermediate result:
    \begin{gather*}
        \phi(K^{(t+1)}) - \phi(K^{*}_{\mathrm{rob}}(\mathcal{C}))
        = (\phi(K^{(t+1)}) - \phi(K^{(t)})) + (\phi(K^{(t)}) - \phi(K^{*}_{\mathrm{rob}}(\mathcal{C}))) \\
        \le -\gamma(K^{(t)}) || g_t ||_{F}^{2} + (\phi(K^{(t)}) - \phi(K^{*}_{\mathrm{rob}}(\mathcal{C}))) \\
        \le \left(1-\gamma(K^{(t)}) /  \mu^{*}(K^{(t)})\right)(\phi(K^{(t)}) - \phi(K^{*}_{\mathrm{rob}}(\mathcal{C}))).
    \end{gather*}
    To then demonstrate the final convergence, we can simply apply this result inductively as follows:
    \begin{gather*}
        \phi(K^{(t)}) - \phi(K^{*}_{\mathrm{rob}}(\mathcal{C}))
        \le \left(1-\gamma(K^{(T)}) / \mu^*(K^{(T)})\right)(\phi(K^{(T-1)}) - \phi(K^{*}_{\mathrm{rob}}(\mathcal{C}))) \\
        \le \left(1-\gamma(K^{(T)})/\mu^*(K^{(T)})\right)\left(1-\gamma(K^{(T-1)})/\mu^*(K^{(T-1)})\right)(\phi(K^{(T-2)}) - \phi(K^{*}_{\mathrm{rob}}(\mathcal{C}))) \\
        \le ...
        \le (\phi(K_{0}) - \phi(K^{*}_{\mathrm{rob}}(\mathcal{C}))) \prod_{t=1}^{T}\left(1-\gamma(K^{(t)})/\mu^*(K^{(t)})\right)
        \le (1-\gamma)^{T} (\phi(K_{0}) - \phi(K^{*}_{\mathrm{rob}}(\mathcal{C}))),
    \end{gather*}
    where we take $\gamma := \min_{t} \gamma(K^{(t)})/\mu^{*}(K^{(t)})$. We now prove $\phi(K^{(t)}) - \phi(K^{(t+1)}) \ge \gamma(K^{(t)}) || g_t ||^{2}_{F}$. 
    To do so, we leverage the result of combining Lemmas 3 and 24 from \citep{fazel2018global}, by which it was demonstrated that for any fixed dynamics $C$, there is a $\beta(C) > 0$ such that $J(K, C) - J(K', C) \ge \beta(C) || \nabla_{K} J(K, C) ||^{2}_{F}$ if $K,K'\in\mathcal{K}(C)$, $X_{K}\succcurlyeq0$, and if $\eta$ satisfies:
    \begin{align*}
        \eta\le\frac{1}{16}\min\left(
            (\frac{\sigma_{\min}(Q)\nu(C)}{J(K, C)})^{2} 
            \frac{1}{|| B || || \nabla_{K} J(K, C) || (1 + || A - BK ||)},
            \frac{\sigma_{\min}(Q)}{2 J(K, C) || R + B^{\top}P_{K}B ||}
        \right),
    \end{align*}
    where $\nu(C) := \min_{K\in\mathcal{K}(C)} \sigma_{\min}(X_{K})$. The stability assumption is satisfied in assuming all iterates $K^{(t)}\in\mathcal{K}(\mathcal{C})$, as $\mathcal{K}(\mathcal{C})\subset\mathcal{K}(C)$. $X_{K}\succcurlyeq0$ is similarly true under the assumption that this property holds for all optimization iterates. The assumption on the learning rate is guaranteed for any $C\in\mathcal{C}$ under the assumption of \Cref{eqn:learning_rate_assump}. To leverage this result, we must, therefore, re-express the quantity of interest into an expression with fixed dynamics:
    \begin{align*}
        \phi(K^{(t)}) - \phi(K^{(t+1)})
        &:= J(K^{(t)}, C^{*}(K^{(t)})) - J(K^{(t+1)}, C^{*}(K^{(t+1)})) \\
        &\ge J(K^{(t)}, C^{*}(K^{(t)})) - J(K^{(t+1)}, C^{*}(K^{(t)})) \\
        &\ge \beta(C^{*}(K^{(t)})) || \nabla_{K} J(K^{(t)}, C^{*}(K^{(t)})) ||^{2}_{F} \\
        &= \beta(C^{*}(K^{(t)})) || g_t ||^{2}_{F}.
    \end{align*}
    Thus, taking $\gamma(K^{(t)}) := \beta(C^{*}(K^{(t)}))$ satisfies the desired property and completes the proof.
\end{proof}

\section{Experimental Controls Setup}\label{section:exp_setup}
As discussed in \Cref{section:related_works}, the standard approach to ``robustness via multiplicative noise'' is non-data-driven specification of the perturbations anticipated upon deployment. They all, however, share the same standard structure of \Cref{eqn:lqrm}, with differences being in the specification of the collection $\{\delta_i\}_{i=1}^{p},\{\gamma_i\}_{i=1}^{q},\{A_i\}_{i=1}^{p}$, and $\{B_i\}_{i=1}^{q}$, where $p=q=2$ is used across experiments. We consider two strategies for the specification of $(\{A_i\},\{B_i\})$ and three for that of $(\{\delta_i\},\{\gamma_i\})$. For the former:
\begin{itemize}
    \item \textbf{Random}
    \begin{itemize}
        \item $A_i[j,k]\sim\mathcal{N}(0,1)$
        \item $B_i[j,k]\sim\mathcal{N}(0,1)$
    \end{itemize}
    \item \textbf{Random Row-Col}
    \begin{itemize}
        \item $A_i[j,:] = A_i[:,k] = 1$ for $j,k\sim\mathrm{Unif}([n])$
        \item $B_i[j,:] = B_i[:,k] = 1$ for $j\sim\mathrm{Unif}([n]),k\sim\mathrm{Unif}([m])$
    \end{itemize}
\end{itemize}
For the latter, the general strategy is to find those $\{\delta_i\}_{i=1}^{p},\{\gamma_i\}_{i=1}^{q}$ that result in unstable dynamics when paired with the corresponding $(\{A_i\},\{B_i\})$ for some choice of controller, which varies across the strategies considered. This in turn defines a problem such that, within some radius of misspecified dynamics that retain stability, the controller still performs well. These methods proceed by initializing $\delta^{(0)}_i = \gamma^{(0)}_i = \textbf{1}$ and iteratively multiplicatively increasing each by some pre-defined factor $\rho$ such that $\delta^{(t)}_i = \rho\delta^{(t-1)}_i$ and similarly for $\gamma^{(t)}$ until $J(A, B, \{A_i\},\{B_i\},\{\delta^{*}_i\},\{\gamma^{*}_i\}, K) = \infty$ in
\begin{align}\label{eqn:exp_setup_obj}
    J(A, B, \{A_i\},\{B_i\},\{\delta_i\},\{\gamma_i\}, K) &:= \int_{0}^{\infty} (x^{\top} Q x + (Kx)^{\top} R (Kx)) dt\\
    \textrm{s.t.} \quad \dot{x} &= \left((A + \sum_{i=1}^{p} \delta_{i} A_i) - (B + \sum_{i=1}^{q} \gamma_{i} B_i)K\right)x \nonumber.
\end{align}
The problem specifications, therefore, vary in the $K$ used as the stopping criterion of \Cref{eqn:exp_setup_obj} and whether $\{\delta^{*}_i\},\{\gamma^{*}_i\}$ are modified in the final specification as follows:
\begin{itemize}
    \item \textbf{Critical}: Consider $K^{(t)} := \argmin_{K} J(A, B, \{A_i\},\{B_i\},\{\delta_i^{(t)}\},\{\gamma_i^{(t)}\}, K)$ in each iterate; Take $\{\delta_i := \delta^{*}_i\},\{\gamma_i := \gamma^{*}_i\}$
    \item \textbf{Open-Loop Mean-Square Stable (Weak)}: Consider $K := \textbf{0}$; Take $\{\delta_i := \nu\delta^{*}_i\},\{\gamma_i := \nu\gamma^{*}_i\}$ for some $\nu\in(0,1)$
    \item \textbf{Open-Loop Mean-Square Unstable}: Consider $K := \textbf{0}$; Take $\{\delta_i := \delta^{*}_i\},\{\gamma_i := \gamma^{*}_i\}$
\end{itemize}

All prediction models $\widehat{f} : \Theta\rightarrow(A,B)$ were multi-layer perceptrons implemented in PyTorch \citep{paszke2019pytorch} with optimization done using Adam \citep{kingma2014adam} with a learning rate of $10^{-3}$ over 1,000 training steps. Training such models required roughly 10 minutes using an Nvidia RTX 2080 Ti GPU for each experimental setup. Running the robust control optimization algorithm took roughly one hour for 1,000 design trials.

\section{Experimental Dynamical Systems Setup}\label{section:exp_system_descriptions}
We consider the following dynamical systems in the experiments. Note that parameters were drawn from normal distributions centered on the nominally reported values from the respective papers these dynamics were considered from.

\subsection{Aircraft Control}\label{section:airfoil_task}
We consider the experimental setup studied in \citep{chrif2014aircraft}, in which optimal control is sought on the deflection angles of an aircraft.
In particular, we assume the dynamics are given by the following:
\begin{gather*}
    A = \begin{bmatrix}
        \gamma_{\beta} & \gamma_{p} & \gamma_{r} & 1 \\
        L_{\beta} & L_{p} & L_{r} & 0 \\
        N_{\beta} & N_{p} & N_{r} & 0 \\
        0 & 1 & 0 & 0
    \end{bmatrix}
    \quad
    B = \begin{bmatrix}
        \gamma_{\delta_{r}} & \gamma_{\delta_{a}} \\
        L_{\delta_{r}} & L_{\delta_{a}} \\
        N_{\delta_{r}} & N_{\delta_{a}} \\
        0 & 0
    \end{bmatrix},
    \qquad \theta := [\gamma, L, N]\in\mathbb{R}^{15}
\end{gather*}

The parameter sampling distributions are given in \Cref{table:aircraft_params}.

\begin{table}[h!]
\centering
\caption{\label{table:aircraft_params} Sampling of parameters for aircraft control task.}
\begin{tabular}{|c|c|c|l|}
\hline
\textbf{Parameter Group} & \textbf{Symbols} & \textbf{Distribution} & \textbf{Hyperparameter Sampling} \\
\hline
$\gamma$ coefficients & $\gamma_{\beta}, \gamma_p, \gamma_r, \gamma_{\delta_r}, \gamma_{\delta_a}$ & $\mathcal{N}(\mu_{\gamma}, \Sigma_{\gamma})$ &
$\mu_{\gamma} \sim \mathcal{U}([0,1]^5)$, $\Sigma_{\gamma} = A A^\top,\ A \sim \mathcal{U}([0,1]^{5\times5})$ \\
\hline
$L$ coefficients & $L_{\beta}, L_p, L_r, L_{\delta_r}, L_{\delta_a}$ & $\mathcal{N}(\mu_L, \Sigma_L)$ &
$\mu_L \sim \mathcal{U}([0,1]^5)$, $\Sigma_L = A A^\top,\ A \sim \mathcal{U}([0,1]^{5\times5})$ \\
\hline
$N$ coefficients & $N_{\beta}, N_p, N_r, N_{\delta_r}, N_{\delta_a}$ & $\mathcal{N}(\mu_N, \Sigma_N)$ &
$\mu_N \sim \mathcal{U}([0,1]^5)$, $\Sigma_N = A A^\top,\ A \sim \mathcal{U}([0,1]^{5\times5})$ \\
\hline
\end{tabular}
\end{table}


\subsection{Load Positioning Control}\label{section:load_pos_task}
We consider the load-positioning system of \citep{ahmadi2023lqr,jiang2016iterative}. In this case, the dynamics are given by:
\begin{gather*}
    A = \begin{bmatrix}
        0 & 1 & 0 & 0 \\
        0 & -\frac{d_{L}}{m_{L}}-\frac{d_{L}}{m_{B}} & \frac{k_{B}}{m_{B}} & \frac{d_{B}}{m_{B}} \\
        0 & 0 & 0 & 1 \\
        0 & \frac{d_{L}}{m_{B}} & -\frac{k_{B}}{m_{B}} & -\frac{d_{B}}{m_{B}}
    \end{bmatrix}
    \quad
    B = \begin{bmatrix}
        0 \\
        \frac{1}{m_{L}} + \frac{1}{m_{B}} \\
        0 \\
        -\frac{1}{m_{B}}
    \end{bmatrix},
    \qquad \theta := [m_{B},m_{L},d_{L},k_{B},d_{B}]\in\mathbb{R}^{5}
\end{gather*}

The parameter sampling distributions are given in \Cref{table:load_pos_params}.

\begin{table}[h!]
\centering
\caption{\label{table:load_pos_params} Sampling of parameters for load positioning task.}
\begin{tabular}{|c|c|c|c|}
\hline
\textbf{Parameter} & \textbf{Symbol} & \textbf{Distribution} & \textbf{Hyperparameter Sampling} \\
\hline
Mass of body       & $m_B$  & $m_B = 1 / u$ & $u \sim \mathcal{U}(0.04,\ 0.0667)$ \\
Mass of load       & $m_L$  & $m_L = 1 / u$ & $u \sim \mathcal{U}(0.3333,\ 1.0)$ \\
Stiffness of body  & $k_B$  & $k_B = u \cdot m_B$ & $u \sim \mathcal{U}(0.4,\ 1.3333)$ \\
Damping of body    & $d_B$  & $d_B = u \cdot m_B$ & $u \sim \mathcal{U}(0.004,\ 0.0667)$ \\
\hline
\end{tabular}
\end{table}


\subsection{Furuta Pendulum}\label{section:furuta_task}
We also consider the Furuta pendulum dynamical system given in \citep{arulmozhi2022kalman}, in which the system dynamics were specified by
\begin{gather*}
    A = \frac{1}{J_T} \begin{bmatrix}
    0 & 0 & J_T & 0 \\
    0 & \frac{1}{4} M_p L_p^2 L_r g & -(J_p + \frac{1}{4} m_p L_p^2) D_r & \frac{1}{2} m_p L_p L_r D_p \\
    0 & -\frac{1}{2} m_p L_p g (J_r + m_p L_r^2) & \frac{1}{2} m_p L_p L_r D_r & -(J_r + m_p L_r^2) D_p
    \end{bmatrix}\qquad
    B = \frac{1}{J_T} \begin{bmatrix}
    0 \\
    0 \\
    J_p + \frac{1}{4} m_p L_p^2 \\
    -\frac{1}{2} m_p L_p L_r
    \end{bmatrix}\\
    \theta := [M_p, m_p, L_p, L_r, J_T, J_p, J_r, D_p, D_r]\in\mathbb{R}^{9}
\end{gather*}

The parameter sampling distributions are given in \Cref{table:pendulum_params}.

\begin{table}[h!]
\centering
\caption{\label{table:pendulum_params} Sampling of parameters for Furuta pendulum task.}
\begin{tabular}{|c|c|c|l|}
\hline
\textbf{Parameter} & \textbf{Symbol} & \textbf{Distribution} & \textbf{Hyperparameter Values} \\
\hline
Pendulum mass & $M_p$ & $|\mathcal{N}(\mu_{M_p}, \sigma_{M_p}^2)|$ & $\mu_{M_p} = 0.024,\ \sigma_{M_p} \sim \mathcal{U}(0,1)$ \\
Rotor mass & $m_p$ & $|\mathcal{N}(\mu_{m_p}, \sigma_{m_p}^2)|$ & $\mu_{m_p} = 0.095,\ \sigma_{m_p} \sim \mathcal{U}(0,1)$ \\
Pendulum length & $L_p$ & $|\mathcal{N}(\mu_{L_p}, \sigma_{L_p}^2)|$ & $\mu_{L_p} = 0.129,\ \sigma_{L_p} \sim \mathcal{U}(0,1)$ \\
Rotor length & $L_r$ & $|\mathcal{N}(\mu_{L_r}, \sigma_{L_r}^2)|$ & $\mu_{L_r} = 0.085,\ \sigma_{L_r} \sim \mathcal{U}(0,1)$ \\
Total inertia & $J_T$ & $|\mathcal{N}(\mu_{J_T}, \sigma_{J_T}^2)|$ & $\mu_{J_T} = f(\mu_{m_p}, \mu_{L_r}, \mu_{J_r}, \mu_{J_p})$, $\sigma_{J_T} \sim \mathcal{U}(0,1)$ \\
Pendulum inertia & $J_p$ & $|\mathcal{N}(\mu_{J_p}, \sigma_{J_p}^2)|$ & $\mu_{J_p} = \frac{M_p L_p^2}{12},\ \sigma_{J_p} \sim \mathcal{U}(0,1)$ \\
Rotor inertia & $J_r$ & $|\mathcal{N}(\mu_{J_r}, \sigma_{J_r}^2)|$ & $\mu_{J_r} = \frac{m_p L_r^2}{12},\ \sigma_{J_r} \sim \mathcal{U}(0,1)$ \\
Pendulum damping & $D_p$ & $|\mathcal{N}(\mu_{D_p}, \sigma_{D_p}^2)|$ & $\mu_{D_p} = 0.0005,\ \sigma_{D_p} \sim \mathcal{U}(0,1)$ \\
Rotor damping & $D_r$ & $|\mathcal{N}(\mu_{D_r}, \sigma_{D_r}^2)|$ & $\mu_{D_r} = 0.0015,\ \sigma_{D_r} \sim \mathcal{U}(0,1)$ \\
\hline
\end{tabular}
\end{table}

\newpage
\subsection{DC Microgrids}\label{section:dc_task}
We additionally consider the LQR model of DC microgrids given in \citep{liu2023novel}, in which the system dynamics were specified by
\begin{gather*}
A = \begin{bmatrix}
\frac{2(-u_0d - NK_2S)}{V_sd} & 0 & \frac{-2NK_4S}{V_sd} & \frac{-4NK_5S}{V_sd} & \frac{2u_0}{V_s} & 0 & 0 & 0 & 0 \\
0 & \frac{2(-u_0d - NK_3S)}{V_sd} & \frac{4NK_4S}{V_sd} & \frac{6NK_5S}{V_sd} & 0 & \frac{2u_0}{V_s} & 0 & 0 & 0 \\
\frac{6NK_2S}{V_sd} & \frac{4NK_3S}{V_sd} & \frac{2(-u_0d - NK_4S)}{V_sd} & 0 & 0 & 0 & \frac{2u_0}{V_s} & 0 & 0 \\
\frac{-4NK_2S}{V_sd} & \frac{-2NK_3S}{V_sd} & 0 & \frac{2(-u_0d - NK_5S)}{V_sd} & 0 & 0 & 0 & \frac{2u_0}{V_s} & 0 \\
\frac{u_0}{V_t} & 0 & 0 & 0 & \frac{-u_0}{V_t} & 0 & 0 & 0 & 0 \\
0 & \frac{u_0}{V_t} & 0 & 0 & 0 & \frac{-u_0}{V_t} & 0 & 0 & 0 \\
0 & 0 & \frac{u_0}{V_t} & 0 &  & 0 & \frac{-u_0}{V_t} & 0 & 0 \\
0 & 0 & 0 & \frac{u_0}{V_t} & 0 & 0 & 0 & \frac{-u_0}{V_t} & 0 \\
\frac{NRT}{FC_2^c} & \frac{-NRT}{FC_3^c} & \frac{NRT}{FC_4^c} & \frac{NRT}{FC_5^c} & 0 & 0 & 0 & 0 & 0
\end{bmatrix}
B = \begin{bmatrix}
\frac{C_2^t - C_2^c}{V_s/2} \\ \frac{C_3^t - C_3^c}{V_s/2} \\ \frac{C_4^t - C_4^c}{V_s/2} \\ \frac{C_5^t - C_5^c}{V_s/2} \\ \frac{C_2^c - C_2^t}{V_t}   \\ \frac{C_3^c - C_3^t}{V_t}   \\ \frac{C_4^c - C_4^t}{V_t}   \\ \frac{C_5^c - C_5^t}{V_t}   \\ 0
\end{bmatrix}\\
\theta := [V_s, V_t, S, d, N, K_2, K_3, K_4, K_5, C^c_2, C^c_3, C^c_4, C^c_5, C^t_2, C^t_3, C^t_4, C^t_5]\in\mathbb{R}^{17}
\end{gather*}

The parameter sampling distributions are given in \Cref{table:dc_params}.

\begin{table}[H]
\centering
\caption{\label{table:dc_params} Sampling of parameters for DC microgrids task.}
\resizebox{\textwidth}{!}{%
\begin{tabular}{|c|c|c|l|}
\hline
\textbf{Parameter} & \textbf{Symbol} & \textbf{Distribution} & \textbf{Hyperparameter Values} \\
\hline
Source voltage           & $V_s$       & $\mathcal{N}(\mu_{V_s}, \sigma_{V_s}^2)$ & $\mu_{V_s} = 40,\ \sigma_{V_s} = 26.67$ \\
Terminal voltage         & $V_t$       & $\mathcal{N}(\mu_{V_t}, \sigma_{V_t}^2)$ & $\mu_{V_t} = 500,\ \sigma_{V_t} = 333.33$ \\
Surface area             & $S$         & $\mathcal{N}(\mu_S, \sigma_S^2)$         & $\mu_S = 24,\ \sigma_S = 16.00$ \\
Diffusion coefficient    & $d$         & $\mathcal{N}(\mu_d, \sigma_d^2)$         & $\mu_d = 1.27\times10^{-3},\ \sigma_d = 8.47\times10^{-4}$ \\
Number of layers         & $N$         & $\mathcal{N}(\mu_N, \sigma_N^2)$         & $\mu_N = 37,\ \sigma_N = 24.67$ \\
Reaction rate constant 2 & $K_2$       & $\mathcal{N}(\mu_{K_2}, \sigma_{K_2}^2)$ & $\mu_{K_2} = 8.768\times10^{-10},\ \sigma_{K_2} = 5.845\times10^{-10}$ \\
Reaction rate constant 3 & $K_3$       & $\mathcal{N}(\mu_{K_3}, \sigma_{K_3}^2)$ & $\mu_{K_3} = 3.222\times10^{-10},\ \sigma_{K_3} = 2.148\times10^{-10}$ \\
Reaction rate constant 4 & $K_4$       & $\mathcal{N}(\mu_{K_4}, \sigma_{K_4}^2)$ & $\mu_{K_4} = 6.825\times10^{-10},\ \sigma_{K_4} = 4.550\times10^{-10}$ \\
Reaction rate constant 5 & $K_5$       & $\mathcal{N}(\mu_{K_5}, \sigma_{K_5}^2)$ & $\mu_{K_5} = 5.897\times10^{-10},\ \sigma_{K_5} = 3.931\times10^{-10}$ \\
Capacitance cell 2 (cathode) & $C^c_2$ & $\mathcal{N}(\mu_{C^c_2}, \sigma_{C^c_2}^2)$ & $\mu_{C^c_2} = 1.0,\ \sigma_{C^c_2} = 0.667$ \\
Capacitance cell 3 (cathode) & $C^c_3$ & $\mathcal{N}(\mu_{C^c_3}, \sigma_{C^c_3}^2)$ & $\mu_{C^c_3} = 1.0,\ \sigma_{C^c_3} = 0.667$ \\
Capacitance cell 4 (cathode) & $C^c_4$ & $\mathcal{N}(\mu_{C^c_4}, \sigma_{C^c_4}^2)$ & $\mu_{C^c_4} = 1.0,\ \sigma_{C^c_4} = 0.667$ \\
Capacitance cell 5 (cathode) & $C^c_5$ & $\mathcal{N}(\mu_{C^c_5}, \sigma_{C^c_5}^2)$ & $\mu_{C^c_5} = 1.0,\ \sigma_{C^c_5} = 0.667$ \\
Capacitance cell 2 (total)   & $C^t_2$ & $\mathcal{N}(\mu_{C^t_2}, \sigma_{C^t_2}^2)$ & $\mu_{C^t_2} = 1.0,\ \sigma_{C^t_2} = 0.667$ \\
Capacitance cell 3 (total)   & $C^t_3$ & $\mathcal{N}(\mu_{C^t_3}, \sigma_{C^t_3}^2)$ & $\mu_{C^t_3} = 1.0,\ \sigma_{C^t_3} = 0.667$ \\
Capacitance cell 4 (total)   & $C^t_4$ & $\mathcal{N}(\mu_{C^t_4}, \sigma_{C^t_4}^2)$ & $\mu_{C^t_4} = 1.0,\ \sigma_{C^t_4} = 0.667$ \\
Capacitance cell 5 (total)   & $C^t_5$ & $\mathcal{N}(\mu_{C^t_5}, \sigma_{C^t_5}^2)$ & $\mu_{C^t_5} = 1.0,\ \sigma_{C^t_5} = 0.667$ \\
\hline
\end{tabular}
}
\end{table}

\newpage
\subsection{Nuclear Plant}\label{section:nuclear_task}
We finally consider the terminal sliding-mode control of a nuclear plant from \citep{kirgni2023lqr}, given by: 
\begin{gather*}
A = 
\begin{bmatrix}
-\frac{\beta}{\Lambda} & \frac{\beta_1}{\Lambda} & \frac{\beta_2}{\Lambda} & \frac{\beta_3}{\Lambda} & \frac{\alpha_f \theta}{\Lambda} & \frac{\alpha_c \theta}{2 \Lambda} & -\frac{\sigma_X \theta}{\nu \Sigma_f \Lambda} & 0 \\
\lambda_1 & -\lambda_1 & 0 & 0 & 0 & 0 & 0 & 0 \\
\lambda_2 & 0 & -\lambda_2 & 0 & 0 & 0 & 0 & 0 \\
\lambda_3 & 0 & 0 & -\lambda_3 & 0 & 0 & 0 & 0 \\
\frac{\epsilon_f P_0}{\mu_f} & 0 & 0 & 0 & -\frac{\Omega}{\mu_f} & \frac{\Omega}{\mu_f} & 0 & 0 \\
\frac{(1-\epsilon_f) P_0}{\mu_c} & 0 & 0 & 0 & \frac{\Omega}{\mu_c} & \frac{2M + \Omega}{2\mu_c} & 0 & 0 \\
(\gamma_X \Sigma_f - \sigma_X X_0)\phi_0 P_0 & 0 & 0 & 0 & 0 & 0 & -\left(\lambda_X + \phi_0 P_0 \theta\right) & \lambda_I \\
\gamma_I \Sigma_f \phi_0 P_0 & 0 & 0 & 0 & 0 & 0 & 0 & -\lambda_I
\end{bmatrix}
\quad
B = 
\begin{bmatrix}
-\frac{\theta}{\Lambda} \\
0 \\
0 \\
0 \\
0 \\
0 \\
0 \\
0
\end{bmatrix}\\
\theta := [\alpha_c, \alpha_f, \beta, \beta_1, \beta_2, \beta_3, \Lambda, \lambda_I, \lambda_X, \lambda_1, \lambda_2, \lambda_3, \mu_f, \mu_c, \gamma_X, \gamma_I, \sigma_X, \Sigma_f, \nu, \epsilon_f, \Omega, M, \theta, P_0, \phi_0, X_0]\in\mathbb{R}^{26}
\end{gather*}

The parameter sampling distributions are given in \Cref{table:fusion_params}.

\begin{table}[H]
\centering
\caption{\label{table:fusion_params} Sampling of parameters for nuclear plant task. Note that some of the parameters, specifically $\theta$, $\mu_c$, $\nu$, $\Omega$, $M$, $\phi_0$, $X_0$, were not ascribed values in the paper from which the dynamics were provided. These were assumed to be in normalized units for the simulation.}
\resizebox{\textwidth}{!}{%
\begin{tabular}{|c|c|c|l|}
\hline
\textbf{Parameter} & \textbf{Symbol} & \textbf{Distribution} & \textbf{Hyperparameter Values} \\
\hline
Coolant reactivity coefficient       & $\alpha_c$     & $\mathcal{N}(\mu_{\alpha_c}, \sigma_{\alpha_c}^2)$ & $\mu = -2.0,\ \sigma = 2.0$ \\
Fuel reactivity coefficient          & $\alpha_f$     & $\mathcal{N}(\mu_{\alpha_f}, \sigma_{\alpha_f}^2)$ & $\mu = -14.0,\ \sigma = 14.0$ \\
Total delayed neutron fraction       & $\beta$        & $\mathcal{N}(\mu_{\beta}, \sigma_{\beta}^2)$       & $\mu = 0.0065,\ \sigma = 0.0065$ \\
Delayed neutron precursor (group 1)  & $\beta_1$      & $\mathcal{N}(\mu_{\beta_1}, \sigma_{\beta_1}^2)$   & $\mu = 0.00021,\ \sigma = 0.00021$ \\
Delayed neutron precursor (group 2)  & $\beta_2$      & $\mathcal{N}(\mu_{\beta_2}, \sigma_{\beta_2}^2)$   & $\mu = 0.00225,\ \sigma = 0.00225$ \\
Delayed neutron precursor (group 3)  & $\beta_3$      & $\mathcal{N}(\mu_{\beta_3}, \sigma_{\beta_3}^2)$   & $\mu = 0.00404,\ \sigma = 0.00404$ \\
Prompt neutron lifetime              & $\Lambda$      & $\mathcal{N}(\mu_{\Lambda}, \sigma_{\Lambda}^2)$   & $\mu = 2.1,\ \sigma = 2.1$ \\
Iodine decay constant                & $\lambda_I$    & $\mathcal{N}(\mu_{\lambda_I}, \sigma_{\lambda_I}^2)$ & $\mu = 10.0,\ \sigma = 10.0$ \\
Xenon decay constant                 & $\lambda_X$    & $\mathcal{N}(\mu_{\lambda_X}, \sigma_{\lambda_X}^2)$ & $\mu = 2.9,\ \sigma = 2.9$ \\
Decay const. neutron prec. group 1   & $\lambda_1$    & $\mathcal{N}(\mu_{\lambda_1}, \sigma_{\lambda_1}^2)$ & $\mu = 0.0124,\ \sigma = 0.0124$ \\
Decay const. neutron prec. group 2   & $\lambda_2$    & $\mathcal{N}(\mu_{\lambda_2}, \sigma_{\lambda_2}^2)$ & $\mu = 0.0369,\ \sigma = 0.0369$ \\
Decay const. neutron prec. group 3   & $\lambda_3$    & $\mathcal{N}(\mu_{\lambda_3}, \sigma_{\lambda_3}^2)$ & $\mu = 0.632,\ \sigma = 0.632$ \\
Fuel heat capacity                   & $\mu_f$        & $\mathcal{N}(\mu_{\mu_f}, \sigma_{\mu_f}^2)$       & $\mu = 0.0263,\ \sigma = 0.0263$ \\
Coolant heat capacity                & $\mu_c$        & $\mathcal{N}(\mu_{\mu_c}, \sigma_{\mu_c}^2)$       & $\mu = 1.0,\ \sigma = 1.0$ \\
Fission yield (xenon)                & $\gamma_X$     & $\mathcal{N}(\mu_{\gamma_X}, \sigma_{\gamma_X}^2)$ & $\mu = 0.003,\ \sigma = 0.003$ \\
Fission yield (iodine)               & $\gamma_I$     & $\mathcal{N}(\mu_{\gamma_I}, \sigma_{\gamma_I}^2)$ & $\mu = 0.059,\ \sigma = 0.059$ \\
Xenon absorption cross-section       & $\sigma_X$     & $\mathcal{N}(\mu_{\sigma_X}, \sigma_{\sigma_X}^2)$ & $\mu = 3.5\times 10^{-18},\ \sigma = 3.5\times 10^{-18}$ \\
Fission cross-section                & $\Sigma_f$     & $\mathcal{N}(\mu_{\Sigma_f}, \sigma_{\Sigma_f}^2)$ & $\mu = 0.3358,\ \sigma = 0.3358$ \\
Neutrons per fission                 & $\nu$          & $\mathcal{N}(\mu_{\nu}, \sigma_{\nu}^2)$           & $\mu = 1.0,\ \sigma = 1.0$ \\
Power deposition fraction in fuel    & $\epsilon_f$   & $\mathcal{N}(\mu_{\epsilon_f}, \sigma_{\epsilon_f}^2)$ & $\mu = 0.92,\ \sigma = 0.92$ \\
Heat transfer coefficient            & $\Omega$       & $\mathcal{N}(\mu_{\Omega}, \sigma_{\Omega}^2)$     & $\mu = 1.0,\ \sigma = 1.0$ \\
Coolant mass                         & $M$            & $\mathcal{N}(\mu_M, \sigma_M^2)$                   & $\mu = 1.0,\ \sigma = 1.0$ \\
Control reactivity                   & $\theta$       & $\mathcal{N}(\mu_{\theta}, \sigma_{\theta}^2)$     & $\mu = 1.0,\ \sigma = 1.0$ \\
Nominal reactor power                & $P_0$          & $\mathcal{N}(\mu_{P_0}, \sigma_{P_0}^2)$           & $\mu = 3.0,\ \sigma = \sqrt{3.0}$ \\
Neutron flux                         & $\phi_0$       & $\mathcal{N}(\mu_{\phi_0}, \sigma_{\phi_0}^2)$     & $\mu = 1.0,\ \sigma = 1.0$ \\
Nominal xenon conc.                  & $X_0$          & $\mathcal{N}(\mu_{X_0}, \sigma_{X_0}^2)$           & $\mu = 1.0,\ \sigma = 1.0$ \\
\hline
\end{tabular}
}
\end{table}

\newpage
\section{Additional Experimental Results}\label{section:addn_exp_results}
\subsection{Raw Results}
We here provide the raw regrets from \Cref{section:experiments_regret} in \Cref{table:regret} and the proportion of stabilized dynamics in \Cref{table:stab_perc}.

\begin{table}[H]
\caption{\label{table:regret} Each of the results below are median normalized regrets over 1,000 i.i.d. test samples with median absolute deviations in parentheses. For clarity, we have not reported any cases with $>80\%$ unstable cases (see \Cref{table:stab_perc} for respective percentages).}
\centering
\resizebox{\textwidth}{!}{%
\begin{tabular}{lccccc}
\toprule
 & Airfoil & Load Positioning & Furuta Pendulum & DC Microgrids & Fusion Plant \\
\midrule
Random Critical & --- & --- & --- & --- & --- \\
Random OL MSS (Weak) & 0.091 (0.045) & --- & --- & --- & --- \\
Random OL MSUS & --- & --- & --- & --- & --- \\
Row-Col Critical & --- & --- & --- & --- & --- \\
Row-Col OL MSS (Weak) & 0.101 (0.063) & --- & --- & --- & --- \\
Row-Col OL MSUS & 0.104 (0.066) & --- & --- & --- & --- \\
CPC & \textbf{0.085 (0.058)} & \textbf{0.033 (0.023)} & \textbf{0.002 (0.002)} & \textbf{0.000 (0.000)} & \textbf{0.011 (0.011)} \\
Shared Lyapunov & 0.349 (0.221) & 0.358 (0.255) & 0.055 (0.039) & 0.000 (0.000) & 0.030 (0.027) \\
Auxiliary Stabilizer & 0.322 (0.202) & 0.343 (0.256) & 0.048 (0.036) & 0.000 (0.000) & 0.029 (0.027) \\
$\mathcal{H}_{\infty}$ & 0.288 (0.188) & 0.087 (0.060) & 0.063 (0.045) & 0.012 (0.010) & 0.035 (0.032) \\
\bottomrule
\end{tabular}
}
\end{table}

\begin{table}[H]
\caption{\label{table:stab_perc} Percentages of cases with unstable robust control over 1,000 i.i.d. test samples.} 
\centering
\resizebox{\textwidth}{!}{%
\begin{tabular}{lccccc}
\toprule
 & Airfoil & Load Positioning & Furuta Pendulum & DC Microgrids & Fusion Plant \\
\midrule
Random Critical & 1.000 & 1.000 & 1.000 & 1.000 & 1.000 \\
Random OL MSS (Weak) & 0.783 & 1.000 & 0.920 & 1.000 & 0.987 \\
Random OL MSUS & 0.825 & 1.000 & 0.961 & 1.000 & 0.990 \\
Row-Col Critical & 0.998 & 1.000 & 1.000 & 1.000 & 1.000 \\
Row-Col OL MSS (Weak) & 0.200 & 1.000 & 0.948 & 1.000 & 0.960 \\
Row-Col OL MSUS & 0.210 & 1.000 & 0.951 & 1.000 & 0.963 \\
CPC & 0.088 & 0.251 & 0.174 & 0.009 & 0.643 \\
Shared Lyapunov & 0.093 & 0.229 & 0.141 & 0.008 & 0.561 \\
Auxiliary Stabilizer & 0.087 & 0.223 & 0.142 & 0.007 & 0.556 \\
$\mathcal{H}_{\infty}$ & 0.081 & 0.236 & 0.142 & 0.007 & 0.570 \\
\bottomrule
\end{tabular}
}
\end{table}

\subsection{Method Timings}\label{section:method_timing}
CPC is more computationally expensive than alternatives. This pairs well with the anticipated use cases, namely in engineering design workflows involving UCCD, i.e. where the control problem is solved \textit{offline}.

\begin{table}[H]
  \caption{Comparison of average method timing (to convergence) across tasks as measured over 10 trials for each experimental setup.}
  \centering
  \begin{tabular}{lccccc}
    \hline
    \textbf{Method} & \textbf{Airfoil} & \textbf{Load Position} & \textbf{Pendulum} & \textbf{Battery} & \textbf{Fusion} \\
    \hline
    $\mathcal{H}_{\infty}$      & 0.17 & 0.14 & 0.16 & 0.19 & 0.23 \\
    Shared Lyapunov            & 2.68 & 2.63 & 2.53 & 1.13 & 2.12 \\
    Auxiliary Stabilizer       & 1.70 & 1.75 & 1.85 & 1.01 & 1.41 \\
    Random Critical            & 7.50 & 1.74 & 6.78 & 7.02 & 10.50 \\
    Random OL MSS (Weak)       & 6.72 & 1.34 & 4.94 & 7.36 & 6.39 \\
    Random OL MSUS             & 6.63 & 2.18 & 7.25 & 15.21 & 7.11 \\
    Row-Col Critical           & 5.37 & 2.32 & 5.51 & 5.48 & 5.44 \\
    Row-Col OL MSS (Weak)      & 4.35 & 2.34 & 4.13 & 1.00 & 2.88 \\
    Row-Col OL MSUS            & 3.43 & 1.84 & 4.77 & 4.54 & 3.18 \\
    CRC                        & 13.03 & 13.44 & 12.47 & 12.19 & 10.88 \\
    \hline
  \end{tabular}
  \label{tab:comparison}
\end{table}

\end{document}